\documentclass[reqno]{amsart}
\usepackage{amsfonts,amsmath,amssymb}

\numberwithin{equation}{section}
\newtheorem{theorem}{Theorem}[section]
\newtheorem{lemma}{Lemma}[section]

\newtheorem{proposition}{Proposition}[section]

\newtheorem{remark}{Remark}[section]

\begin{document}
\title[Euler-Poisson: Ion equation]{Global Smooth Ion Dynamics in the
Euler-Poisson System}
\author{Yan Guo and Benoit Pausader}

\begin{abstract}
A fundamental two-fluid model for describing dynamics of a plasma is the
Euler-Poisson system, in which compressible ion and electron fluids interact
with their self-consistent electrostatic force. Global smooth electron
dynamics were constructed in Guo \cite{Guo} due to dispersive effect of the
electric field. In this paper, we construct global smooth irrotational
solutions with small amplitude for ion dynamics in the Euler-Poisson system.
\end{abstract}

\address{ }
\email{guoy@cfm.brown.edu Benoit.Pausader@math.brown.edu}
\maketitle
\tableofcontents

\section{Introduction and Formulation}

\subsection{The ionic Euler-Poisson system}

The \textquotedblleft two-fluid\textquotedblright\ models in plasma physics
describe dynamics of two separate compressible fluids of ions and electrons
interacting with their self-consistent electromagnetic field. Many famous
nonlinear dispersive PDE, such as Zakharov's equation, nonlinear Schr\"{o}dinger equations, as well as KdV equations, can be formally derived from
\textquotedblleft two-fluid\textquotedblright\ models under various
asymptotic limits. In the absence of the magnetic effects, the fundamental
two-fluid model for describing the dynamics of a plasma is given by the
following Euler-Poisson system 
\begin{equation}
\begin{split}
\partial _{t}n_{\pm }+\nabla \cdot \left( n_{\pm }v\,_{\pm }\right) & =0 \\
n_{\pm }m_{\pm }(\partial _{t}v_{\pm }+v_{\pm }\cdot \nabla v_{\pm })+T_{\pm
}\nabla n_{\pm }& =en_{\pm }\nabla \phi \\
\Delta \phi & =4\pi e(n_{+}-n_{-}).
\end{split}
\label{2Fluids}
\end{equation}
Here $n_{\pm }$ are the ion (+) and electron density (-), $v_{\pm }$ are the
ion (+) and electron (-) velocity, $m_{\pm }$ are the masses of the ions (+)
and electrons (-), $T_{\pm }$ are their effective temperatures, and $e$ is
the charge of an electron. The self-consistent electric field $\nabla \phi $
satisfies the Poisson equation. The Euler-Poisson system describes rich
dynamics of a plasma. Indeed, even at the linearized level, there are
electron waves, ion acoustic waves in the Euler-Poisson system. Despite its
importance, there has been few mathematical study of its global solutions in
3D. This stems from the fact that the Euler-Poisson system belongs to the
general class of hyperbolic conservation laws with zero dissipation, for
which no general mathematical framework for construction of global in-time
solutions exists in 3D. In fact, as expected \cite{GuoTad}, solutions of the
Euler-Poisson system with large amplitude in general will develop shocks.

However, unlike the pure Euler equations, shock formation for solutions of
the Euler-Poisson system with small amplitude has remained open. In Guo \cite
{Guo}, the first author studied a simplified model of the Euler-Poisson
system for an electron fluid: 
\begin{equation}
\begin{split}
\partial _{t}n_{-}+\nabla \cdot \left( n_{-}v\,_{-}\right) & =0 \\
n_{-}m_{-}(\partial _{t}v_{-}+v_{-}\cdot \nabla v_{-})+T_{-}\nabla n_{-}&
=en_{-}\nabla \phi  \\
\Delta \phi & =4\pi e(n_{-}-n_{0}).
\end{split}
\label{electronEP}
\end{equation}
In this model, the ions are treated as immobile and only form a constant
charged background $n_{0}$. Surprisingly, it was observed \cite{Guo} that
the linearized Euler-Poisson system for the electron fluid is the
Klein-Gordon equation, due to plasma oscillations created by the electric
field $\phi .$ In this case, the dispersion relation reads 
\begin{equation*}
\omega (\xi )\backsim \sqrt{1+|\xi |^{2}}
\end{equation*}%
Such a \textquotedblleft Klein-Gordon\textquotedblright\ effect led to
construction of smooth irrotational electron dynamics with small amplitude
for all time. This is in stark contrast to the pure Euler equations for
neutral fluids where the dispersion relation reads 
\begin{equation*}
\omega (\xi )\backsim |\xi |,
\end{equation*}
in which shock waves can develop even for small smooth initial data (see
Sideris \cite{Sid}). It is the dispersive effect of the electric field that
enhances the linear decay rate and prevents shock formation. The natural
open question remains: does such a dispersive effect exist generally ? If
so, can it prevent shock formation for the general Euler-Poisson system
\eqref{2Fluids} ?

In the current paper, we make another contribution towards answering this
question. We consider another (opposite) asymptotic limit of the original
Euler-Poisson system \eqref{2Fluids} for the ion dynamics. It is well-known
that $\frac{m_{-}}{m_{+}}<<1$ in all physical situations. By letting the
electron mass $m_{-}$ go to zero, we formally obtain $T_{-}\nabla
n_{-}=en_{-}\nabla \phi $ and the famous Boltzmann relation 
\begin{equation}
n_{-}=n_{0}\exp (\frac{e\phi }{T_{-}})  \label{boltzmann}
\end{equation}%
for the electron density ($n_{0}$ is a constant). Such an important relation %
\eqref{boltzmann} can also be verified through arguments from kinetic
theory, see Cordier and Grenier \cite{CorGre}. We then obtain the well-known
ion dynamic equations as 
\begin{equation}
\begin{split}
\partial _{t}n_{+}+\nabla \cdot \left( n_{+}v_{+}\right) & =0 \\
n_{+}m_{+}\left( \partial _{t}v_{+}+v_{+}\cdot \nabla v_{+}\right) &
=-T_{+}\nabla n_{+}-n_{+}e\nabla \phi \\
\Delta \phi & =4\pi e\left( n_{0}\exp \left( \frac{e\phi }{T_{-}}\right)
-n_{+}\right) .
\end{split}
\label{EP}
\end{equation}%
We also assume that 
\begin{equation}
\hbox{curl}(v(0))=0.  \label{CondCurlFree}
\end{equation}%
It is standard that the condition \eqref{CondCurlFree} is preserved by the
flow. As a matter of fact, non irrotational flow leads to creation of a
non-vanishing magnetic field, which is omitted in the Euler-Poisson system
but retained in a more general Euler-Maxwell system \cite{CheJerWan}. The
linear dispersion relation for \eqref{EP} behaves like 
\begin{equation}
p(\xi )\equiv |\xi |\sqrt{\frac{2+|\xi |^{2}}{1+|\xi |^{2}}}\equiv |\xi
|q(|\xi |)  \label{DefOfp}
\end{equation}%
which is much closer to the wave dispersion $\omega (\xi )=|\xi |$ than to
the Klein-Gordon one, $\omega (\xi )=\sqrt{1+|\xi |^{2}}$ (in particular
note this dispersion relation behaves near $0$ as in the Schr\"{o}dinger
case, whereas in our dispersion relation $p$ remains very similar to that of
the wave's). Intuitively, one might expect formation of singularity for %
\eqref{EP} as in the pure Euler equations. Nevertheless, we demonstrate that
small smooth irrotational flows exist globally in time, and there is no
shock formation. Without loss of generality, we study the global behavior of
irrotational perturbations of the uniform state 
\begin{equation*}
\lbrack n_{+},v_{+}]=[n_{0}+\rho ,v].
\end{equation*}%
We use two important norms defined as follows: 
\begin{equation}
\begin{split}
\Vert u(x)\Vert _{Y}& =\Vert |\nabla |^{-1}u\Vert _{H^{2k+1}}+\Vert u\Vert
_{W^{k+\frac{12}{5},\frac{10}{9}}} \\
\Vert u(t,x)\Vert _{X}& =\sup_{t}\left( \Vert |\nabla |^{-1}(1-\Delta )^{k+%
\frac{1}{2}}u(t)\Vert _{L^{2}}+(1+t)^{\frac{16}{15}}\Vert (1-\Delta )^{\frac{%
k}{2}}u(t)\Vert _{L^{10}}\right)
\end{split}
\label{Norm}
\end{equation}%
for $k\geq 5$.

Here, we keep $k$ as a parameter to emphasize the fact that smoother initial
data lead to smoother solutions. Hidden in the $X$-norm is a statement about
preservation of regularity of $(\rho ,v)$. Our main result is the following

\begin{theorem}
\label{MainThm} There exists $\varepsilon >0$ such that any initial
perturbation $(n_{0}+\rho _{0},v_{0})$ satisfying \eqref{CondCurlFree} $%
\nabla \times v_{0}=0$ and $\Vert \rho_0 \Vert _{Y}+\Vert v_{0}\Vert
_{Y}\leq \varepsilon $ leads to a global solution $(n_0+\rho ,v)$ of %
\eqref{EP} with 
\begin{equation*}
\Vert\rho \Vert _{X}+\Vert v\Vert _{X}\leq 2\varepsilon .
\end{equation*}
In particular, the perturbations $\rho $ and $v$ decay in $L^{\infty }$.
\end{theorem}

Together with earlier result in Guo \cite{Guo}, global smooth potential
flows with small velocity exist for two opposite scaling limits of %
\eqref{2Fluids}. This is a strong and exciting indication that shock waves
of small amplitude should be absent for the full Euler-Poisson system %
\eqref{2Fluids}, at least in certain physical regimes. Our method developed
in this paper should be useful in the future study of \eqref{2Fluids}.

There have been a lot of mathematical studies of various aspects of the
Euler-Poisson system for a plasma. Texier \cite{Tex,Tex2} studied the Euler-Maxwell system and its approximation by the Zakharov equations. Wang and Wang \cite{WanWan} constructed large BV
radially symmetric solutions outside the origin. In Liu and
Tadmor \cite{LiuTad,LiuTad2}, threshold for singularity formation has been
studied for the Euler-Poisson system with $T_{\pm }=0$ in one and two
dimensions. In Feldman, Ha and Slemrod \cite{FelSeuSle,FelSeuSle2}, plasma
sheath problem of the Euler-Poisson system was investigated. In Peng and
Wang \cite{PenWan}, Euler-Poisson system is derived from Euler-Maxwell
system with a magnetic field. Quasi-neutral limit in the Euler-Poisson
system was studied in Cordier and Grenier \cite{CorGre} and Peng and Wang 
\cite{PenWan2}. When $n_{+}$ is replaced by a doping profile and a momentum
relaxation is present, the Euler-Poisson system describes electron dynamics
in a semiconductor device. There has been much more mathematical study of
such a model, for which we only refer to Chen, Jerome and Wang \cite%
{CheJerWan} and the references therein.

\subsection{Presentation of the paper}

For notational simplicity, we let $n_{0}=e=T_{+}=T_{-}=1$ in \eqref{EP}
throughout the paper. Even though the ion dynamics system \eqref{EP} is the
most natural system to further understand the dispersive effects in the full
Euler-Poisson system \eqref{2Fluids}, it has been remained an open problem
to construct global smooth solutions until now ever since the work of \cite%
{Guo}, due to much more challenging mathematical difficulties than in the
case of the electron Euler-Poisson equation \eqref{electronEP} studied by
Guo \cite{Guo}.

The first difficulty is to understand the time decay rate of the linearized
ion dynamics equation: 
\begin{equation*}
\partial _{tt}\rho -\Delta \rho -\Delta (-\Delta +1)^{-1}\rho =0.
\end{equation*}%
whose solutions are given by the operator $e^{\pm p(|\nabla |)t}$ with p
given by \eqref{DefOfp}. Unlike the linearized electron equations studied in 
\cite{Guo}, there is no direct study of the linear decay of such a system.
Only recently \cite{GuoPenWan}, time-decay rate for general dispersive
equations has been carried out in detail with asymptotic conditions near low
frequency $|\xi |=0$ and high frequency $|\xi |=\infty .$ Interestingly, any
phase $p(\xi )$ which is not exactly the phase function of the wave equation 
$p(\xi )=|\xi |$ commands a decay rate better than $\frac{1}{t}$. We are
able to employ this result together with a stationary phase analysis near
the inflection point of $p(\xi )$ to obtain a decay rate of $\frac{1}{t^{4/3}%
}$, which is between the wave and the Klein-Gordon equations. A consequence
of the linear estimates of Section \ref{SecLinEst} is that 
\begin{equation*}
\Vert e^{itp(|\nabla |)}\alpha _{0}\Vert _{X}\lesssim \Vert \alpha _{0}\Vert
_{Y}.
\end{equation*}

The main mathematical difficulty in this paper stems from bootstrapping the
linear decay into a construction of global solutions to the nonlinear
problem. Based on very recent new techniques of harmonic analysis in the
study of dispersive PDE by Germain, Masmoudi and Shatah \cite%
{GerMasSha,GerMasSha2,GerMasSha3}, Gustafson, Nakanishi and Tsai \cite%
{GNT2,GNT}, Shatah \cite{Sha}, we follow a new set-up for normal form
transformation in \cite{GerMasSha,GNT,Sha}. Using that $\nabla \times
v\equiv 0$, we can introduce a pair of complex valued new unknowns: 
\begin{equation}
\alpha _{1}=\rho -\frac{i}{q(|\nabla |)}\mathcal{R}^{-1}v,\text{ and }\alpha
_{2}=\rho +\frac{i}{q(|\nabla |)}\mathcal{R}^{-1}v,  \label{DefOfAlpha}
\end{equation}%
for $q$ defined in \eqref{DefOfp}, where $\mathcal{R}=\nabla |\nabla |^{-1}$
stands for the Riesz transform, and $v=\nabla \psi $, $\mathcal{R}%
^{-1}v\equiv |\nabla |\psi $. After the normal form transformation %
\eqref{alpha}, it suffices for us to control 
\begin{equation}
\begin{split}
& \hat{\alpha}(t)\backsim \int_{\mathbb{R}^{3}}\frac{m(\xi ,\eta )}{\Phi
_{1}(\xi ,\eta )}\hat{\alpha}(\xi -\eta )\hat{\alpha}(\eta )d\eta \\
& +\int_{0}^{t}\int_{\mathbb{R}^{6}}e^{i(t-s)p(|\xi |)}\frac{m(\xi ,\eta
)m(\eta ,\zeta )}{\Phi _{1}(\xi ,\eta )}\hat{\alpha}(\xi -\eta )\hat{\alpha}%
(\eta -\zeta )\hat{\alpha}(\zeta )d\eta d\zeta ds.
\end{split}
\label{normal}
\end{equation}%
where $m$ denotes a generic multiplier given by \eqref{DefOfPhim}. This is
well defined in views of \eqref{CondCurlFree} and 
\begin{equation*}
\Phi _{1}=p(|\xi |)-p(|\xi -\eta |)-p(|\eta |).
\end{equation*}%
In the Klein-Gordon case, the phase is bounded away from zero so there is no
singularity. However, for $\Phi _{1},$ there is a significant zero set when $%
|\xi -\eta ||\eta |=0,$ (see Lemma \ref{EstimPhiGen}) and there is no
\textquotedblleft null form\textquotedblright\ structure to cancel with the
multiplier $m$. We first observe that 
\begin{equation*}
m(\xi ,\eta )m(\eta ,\zeta )\backsim |\xi ||\eta |.
\end{equation*}%
We then make use of such a structure to form a locally bounded multiplier 
\begin{equation*}
\mathcal{M}_{1}=\frac{|\xi ||\xi -\eta ||\eta |}{\Phi _{1}(\xi ,\eta )}%
\lesssim 1.
\end{equation*}%
This process introduces a singular term $\frac{\hat{\alpha}(\xi -\eta )}{%
|\xi -\eta |},$ which will be controlled in a separate fashion by the $%
H^{-1} $ norm in our norm $\Vert \cdot \Vert _{X}.$ We believe that
including this $H^{-1}$ control in the norm should work equally well for
equations with nonlinearity which has perfect spatial derivatives.

Even though $\mathcal{M}_{1}$ is locally bounded, it is very difficult to
employ classical bilinear estimate such as Coifman-Meyer Theorem \cite%
{CoiMey} to control \eqref{normal}. This is due to the anisotropic nature of 
$\mathcal{M}_{1}$ since $|\eta |$ can be very small with respect to $|\xi
-\eta |$. Instead, we make use of a very recent multiplier estimate by
Gustafson, Nakanishi and Tsai \cite{GNT}. It is important to use $L^{10}$%
-norms as a proxy for the $L^{\infty }$-norm for which our degenerate
multipliers are not well-suited (we need an $L^{p}$-norm with $p<12$). The
optimal Sobolev regularity for $\mathcal{M}_{1}\in L_{\xi }^{\infty }(\dot{H}%
_{\eta }^{5/4-\varepsilon })\cap L_{\eta }^{\infty }(\dot{H}_{\xi
}^{5/4-\varepsilon })$ is crucial in applying such an estimate to obtain $%
L^{10}$ decay, and its proof is particularly delicate for small frequencies.
We split the phase space and make a careful interplay between angles and the
lengths of $\xi ,\eta ,\xi -\eta $. We also make use of Littlewood-Paley
decomposition and interpolation to obtain a sharp Sobolev estimate for $%
\mathcal{M}_{1}$. On the other hand, to reduce the requirement of number of
derivatives in our norm $X,$ we also need to show a stronger estimate $%
\mathcal{M}_{1}\in L_{\xi }^{\infty }(\dot{H}_{\eta }^{3/2-\varepsilon
})\cap L_{\eta }^{\infty }(\dot{H}_{\xi }^{3/2-\varepsilon })$ for large
frequencies.

\medskip

This paper is organized as follows: in Section \ref{SecLinEst} we study the
relevant linear dispersive equation. In Section \ref{SecNot}, we introduce
our normal form transformation. In Section \ref{SecH-1}, we get an estimate
on the $L^2$-part of the $X$ norm using the energy method. In Section \ref%
{SecMult} we state and prove the relevant multiplier estimate we need in
order to control our bilinear terms. Finally, in Section \ref{SecL10}, we
control the high integrability part of the norm, and finish the analysis to
obtain global solutions with small initial data in Theorem \ref{MainThm}.

\subsection{Notations and preliminary results}

We work in dimension $n=3$, although we state some results in arbitrary
dimension $n$. We introduce 
\begin{equation*}
\langle a\rangle =\sqrt{1+a^{2}}.
\end{equation*}
We write $A\lesssim B$ to signify that there exists a constant $C$ such that 
$A\le C$. We write $A\simeq B$ if $A\lesssim B\lesssim A$. Our phases and
some multiplier are radial functions, and in some cases we might abuse
notations and write, for a radial function $f$, $f(x)=f(\vert x\vert)$.

\medskip

Our multipliers are estimated using the homogeneous Sobolev norm defined for 
$0\le s<n/2$ by 
\begin{equation*}
\Vert f\Vert_{\dot{H}^s}=\Vert \vert\nabla\vert^sf\Vert_{L^2}
\end{equation*}
where $\vert\nabla\vert$ is defined by $\mathcal{F}(\vert\nabla\vert
f)(\xi)=\vert\xi\vert\hat{f}(\xi)$.

\medskip

We will also use the Littlewood-Paley multipliers $P_{N}$ defined for dyadic
numbers $N\in 2^{\mathbb{Z}}$ by 
\begin{equation}
P_{N}g=\mathcal{F}_{\xi }^{-1}\varphi (\frac{\xi }{N})\mathcal{F}_{\xi }g
\label{DefLitPalOp}
\end{equation}
where $\varphi \in C_{c}^{\infty }(\mathbb{R}^{n})$ is such that 
\begin{equation*}
\forall \xi \neq 0,\hskip.2cm\sum_{N\in 2^{\mathbb{Z}}}\varphi (\frac{\xi }{%
N })=1,
\end{equation*}
and for later use, we also introduce a function $\chi\in C^\infty_c(\mathbb{R%
}^n)$ such that $\chi\varphi=\varphi$. An important estimate on these
Littlewood-Paley multipliers is the Bernstein inequality: 
\begin{equation}  \label{BernSobProp}
\begin{split}
&\hskip.8cm\Vert \vert\nabla\vert^{\pm s}P_Nf\Vert_{L^p}\lesssim_s N^{\pm
s}\Vert P_Nf\Vert_{L^p}\lesssim_s N^{\pm s}\Vert f\Vert_{L^p} \\
\end{split}%
\end{equation}
for all $s\ge 0$, and all $1\le p\le\infty$, independently of $f$, $N$, and $%
p$, where $\vert\nabla\vert^s$ is the classical fractional differentiation
operator.

We will also need the two following product estimates:

\begin{lemma}
\label{LemProdEnergy} Let $\tau$ be a multi-index of length $\vert\tau\vert$
and $\gamma<\tau$, then for all $u\in C^\infty_c(\mathbb{R}^n)$ and all $%
\delta>0$, there holds that 
\begin{equation*}
\Vert D^{\tau-\gamma}u D^\gamma\partial_ju\Vert_{L^2}\lesssim_\delta \Vert
u\Vert_{W^{1+\delta,\infty}}\Vert u\Vert_{H^{\vert\tau\vert}}\lesssim \Vert
u\Vert_{W^{2,10}}\Vert u\Vert_{H^{\vert\tau\vert}}
\end{equation*}
\end{lemma}

\begin{proof}
We first note that without loss of generality, we may assume that $|\gamma
|+1,|\tau |-|\gamma |\geq 2$, otherwise H\"{o}lder's inequality gives the
result. We use a simple paradifferential decomposition, in other words, we
write 
\begin{equation*}
\begin{split}
D^{\tau -\gamma }uD^{\gamma }\partial _{j}u& =\left( \sum_{M\sim
N}+\sum_{M/N\leq 1/16}+\sum_{N/M\leq 1/16}\right) \left( P_{M}D^{\tau
-\gamma }u\right) \left( P_{N}D^{\gamma }\partial _{j}u\right) \\
& =R+T_{1}+T_{2}
\end{split}%
\end{equation*}%
where $M$ and $N$ are dyadic numbers. We first estimate $R$ as follows using
Bernstein properties and in particular the fact that 
\begin{equation*}
\Vert P_{N}u\Vert _{L^{\infty }}\lesssim \min (1,N^{-1-\delta })\Vert u\Vert
_{W^{1+\delta ,\infty }}
\end{equation*}%
we get that with the Cauchy Schwartz inequality that 
\begin{equation*}
\begin{split}
\Vert R\Vert _{L^{2}}& \lesssim \sum_{M\sim N}\Vert P_{M}D^{\tau -\gamma
}uP_{N}D^{\gamma }\partial _{j}u\Vert _{L^{2}} \\
& \lesssim \sum_{M\sim N}M^{|\tau |-|\gamma |}\Vert P_{M}u\Vert
_{L^{2}}M^{|\gamma |+1}\Vert P_{N}u\Vert _{L^{\infty }} \\
& \lesssim \left( \sum_{M}M^{2|\tau |}\Vert P_{M}u\Vert _{L^{2}}^{2}\right)
^{\frac{1}{2}}\left( \sum_{M}M^{2}\Vert P_{M}u\Vert _{L^{\infty
}}^{2}\right) ^{\frac{1}{2}} \\
& \lesssim \Vert u\Vert _{H^{|\tau |}}\Vert u\Vert _{W^{1+\delta ,\infty }}.
\end{split}%
\end{equation*}%
Independently, we estimate $T_{1}$ as follows using that if $16M_{i}\leq
N_{i}$, $i=1,2$ then 
\begin{equation*}
\langle P_{M_{1}}fP_{N_{1}}g,P_{M_{2}}hP_{N_{2}}k\rangle _{L^{2}\times
L^{2}}=0
\end{equation*}%
unless $N_{1}\leq 4N_{2}\leq 16N_{1}$ (intersection of the Fourier support),
and letting $f=D^{\tau -\gamma }u$, $g=D^{\gamma }\partial _{j}u$, we get 
\begin{equation*}
\begin{split}
\Vert T_{1}\Vert _{L^{2}}& \lesssim \sum_{N_{1}\sim N_{2},16M_{i}\leq
N_{i}}\langle P_{N_{1}}fP_{M_{1}}g,P_{N_{2}}fP_{M_{2}}g\rangle _{L^{2}\times
L^{2}} \\
& \lesssim \Vert u\Vert _{W^{1,\infty }}^{2}\sum_{N_{1}\sim
N_{2},16M_{i}\leq N_{i}}\Vert P_{N_{1}}f\Vert _{L^{2}}\Vert P_{N_{2}}f\Vert
_{L^{2}}(M_{1}M_{2})^{|\gamma |} \\
& \lesssim \Vert u\Vert _{W^{1,\infty }}^{2}\sum_{N_{1}\sim
N_{2}}N_{1}^{|\gamma |}\Vert P_{N_{1}}f\Vert _{L^{2}}N_{2}^{|\gamma |}\Vert
P_{N_{2}}f\Vert _{L^{2}} \\
& \lesssim \Vert u\Vert _{W^{1,\infty }}^{2}\Vert u\Vert _{H^{|\tau |}}^{2}
\end{split}%
\end{equation*}%
and $T_{2}$ is treated exactly in the same way.
\end{proof}

We also need the following ``tame'' product estimate (see e.g. Tao \cite%
{TaoBook})

\begin{lemma}
For $1<p<\infty$, $s\ge 0$, 
\begin{equation}  \label{tameEst}
\Vert uv\Vert_{W^{s,p}}\lesssim \Vert u\Vert_{L^\infty}\Vert
v\Vert_{W^{s,p}}+\Vert u\Vert_{W^{s,p}}\Vert v\Vert_{L^\infty}
\end{equation}
for $u$ and $v$ in $L^\infty\cap W^{s,p}$.
\end{lemma}

%\begin{lemma}
%\label{TameLemma} For smooth functions $u$, $v$, there holds that 
%\begin{equation*}
%\Vert u v\Vert_{W^{s,\frac{5}{2}}}\lesssim \left(\Vert u\Vert_{H^{s+\frac{3}{
%10}}}\Vert v\Vert_{W^{1,10}}+ \Vert v\Vert_{H^{s+\frac{3}{10}}}\Vert
%u\Vert_{W^{1,10}}\right)
%\end{equation*}
%\end{lemma}

%\begin{proof}
%Indeed, using the ``tame estimate'', we get that 
%\begin{equation*}
%\begin{split}
%\Vert uv\Vert_{W^{s,\frac{5}{2}}}&\lesssim \Vert u\Vert_{W^{s,\frac{5}{2}
%}}\Vert v\Vert_{L^\infty}+\hbox{sym} \\
%&\lesssim \Vert u\Vert_{H^{s+\frac{3}{10}}}\Vert v\Vert_{W^{1,10}}+\hbox{sym}
%\end{split}%
%\end{equation*}
%where ``sym'' means that we have a similar term with $u$ and $v$ exchanged.
%\end{proof}

%%%%%%%%%%%%%%%%%%%%%%%%%%%%%%%%%%%%%%%%%%%%%%%%%%%%%%%%%%%%%%%%%%%%%%%%%%%%%%%%%%%%%%%%%%%%%%%%%%%%%%%%%%%%%%%%%%%%%%%%%%%%%%%%%%%%%%%%%%%%%%%%%%%%%%%%%%%%%%%%%%%%%%%%%%%%%%%%%%%%%%%%%%%%%%%%%%%%%%%%%%%%%%%%%%%%%%%%%%%%%%%%%%%%%%%%%%%%%%%%%%%%%%%%%%%%%%%%%%%%%%%%%%%%%%%%%%%%%%%%%%%%%

\section{Linear Decay}

\label{SecLinEst}

In this section, we investigate the decay of linear solutions of the
linearized equation 
\begin{equation}  \label{LinEqt1}
\partial_{tt}\rho-\Delta\rho-\Delta(-\Delta+1)^{-1}\rho=0.
\end{equation}
These solutions can be expressed in terms of the initial data and of one
``half-wave'' operator 
\begin{equation*}
T_t=e^{itp(\vert\nabla\vert)}
\end{equation*}
for $p$ defined in \eqref{DefOfp} that we now study. Our main result in this
section is the following

\begin{proposition}
\label{PropLinDecay} For any $\delta>0$, for any $f\in W^{\frac{5}{2}%
+\delta,1}$, there holds that 
\begin{equation}  \label{LinDecay}
\Vert e^{itp(\vert\nabla\vert)}f\Vert_{L^\infty}\lesssim_\delta \left(\vert
t\vert^{-\frac{4}{3}}+\vert t\vert^{-\frac{3}{2}}\right)\Vert f\Vert_{W^{%
\frac{5}{2}+\delta,1}}
\end{equation}
for all $t\ne 0$. Besides, we have the $L^{10}$-decay estimate 
\begin{equation}  \label{L10Decay}
\Vert e^{itp(\vert\nabla\vert)}f\Vert_{L^{10}}\lesssim (1+\vert t\vert)^{-%
\frac{16}{15}}\Vert f\Vert_{W^{\frac{12}{5},\frac{10}{9}}}
\end{equation}
uniformly in $\varepsilon, t, f$.
\end{proposition}

More precise estimates are derived below. The rest of the section in devoted
to a proof of \eqref{LinDecay} and \eqref{L10Decay}.

\medskip

For most of this section, we study the dispersive features of our operator
in general dimension $n$. Proposition \ref{PropLinDecay} is a consequence of
the particular case $n=3$. Direct computations give that 
\begin{equation}  \label{Estimp}
\begin{split}
p^{\prime }(r)& =\frac{1}{\sqrt{(1+r^{2})(2+r^{2})}}\left( 1+r^{2}+\frac{1}{
1+r^{2}}\right) , \\
p^{\prime \prime }(r)& =\frac{r\left( r^{4}-2r^{2}-6\right) }{(1+r^{2})\left[
(1+r^{2})(2+r^{2})\right] ^{\frac{3}{2}}}\hskip.1cm\hbox{and} \\
p^{\prime \prime \prime }(r)& =\frac{5r^{4}-6r^{2}-6}{(1+r^{2})^{\frac{5}{2}
}(2+r^{2})^{\frac{3}{2}}}-\frac{r\left( r^{4}-2r^{2}-6\right) (11r^{3}+16r)}{
(1+r^{2})^{\frac{7}{2}}(2+r^{2})^{\frac{5}{2}}}.
\end{split}%
\end{equation}
We note that $p^{\prime \prime }(r)$ has one unique positive root at 
\begin{equation}
r=r_{0}=\sqrt{1+\sqrt{7}}.  \label{r0}
\end{equation}

In order to state our first result, we define a frequency localization
function around the critical point $r_{0}$. Let $\psi _{r_{0}}\in C^{\infty
}(\mathbb{R})$ be a smooth function such that $0\leq \psi \leq 1$, $\psi
_{r_{0}}(r_{0}+r)=1$ when $|r|\leq \varepsilon $ and $\psi
_{r_{0}}(r_{0}+r)=0$ when $|r|\geq 2\varepsilon $.

\begin{lemma}
\label{LinearDecayMedFreqLemma} For all time $t\neq 0$, and all $f\in L^{1}$%
, there holds that 
\begin{equation}
\Vert e^{itp(|\nabla |)}\psi _{r_{0}}(|\nabla |)f\Vert _{L^{\infty
}}\lesssim _{n,\varepsilon }(1+|t|)^{-\frac{n-1}{2}-\frac{1}{3}}\Vert f\Vert
_{L^{1}}.  \label{LinearDecay1}
\end{equation}
\end{lemma}

\begin{proof}
We note that 
\begin{eqnarray*}
\Vert e^{itp(|\nabla |)}\psi _{r_{0}}(|\nabla |)f(x)\Vert_{\infty } &=&||%
\mathcal{F}^{-1}\{e^{itp(|\xi |)}\psi _{r_{0}}(|\xi |)\hat{f}(\xi
)\}||_{\infty } \\
&=&||\mathcal{F}^{-1}\{e^{itp(|\xi |)}\psi _{r_{0}}(|\xi |)\}\ast
f(x)||_{\infty } \\
&\leq &||\mathcal{F}^{-1}\{e^{itp(|\xi |)}\psi _{r_{0}}(|\xi |)\}||_{\infty
}||f||_{L^{1}}.
\end{eqnarray*}
Since $\psi _{r_{0}}$ is chosen to be spherically symmetric, it is
well-known that 
\begin{eqnarray*}
\mathcal{F}^{-1}\{e^{itp(|\xi |)}\psi _{r_{0}}(|\xi |)\}(x) &=&2\pi
\int_{0}^{\infty }e^{itp(r)}\psi_{r_0} (r)\tilde{J}_{\frac{n-2}{2}%
}(r|x|)r^{n-1}dr \\
&=&2\pi \int_{0}^{\infty }e^{itp(r)}\psi_{r_0} (r)\tilde{J}_{\frac{n-2}{2}
}(r|x|)r^{n-1}dr
\end{eqnarray*}
where for all $n\geq 2$, 
\begin{equation}
\tilde{J}_{\frac{n-2}{2}}(s)\equiv s^{-\frac{n-2}{2}}J_{\frac{n-2}{2}}(s)= %
\hbox{Re}\left( e^{is}Z(s)\right) =e^{is}Z(s)-e^{-is}\bar{Z}(s).
\label{Bessel1}
\end{equation}
Here $Z(s)$ is a smooth function satisfying (cf John \cite{Joh}) that for
all $k\geq 0$ and all $s$ 
\begin{equation}
|\partial ^{k}Z(s)|\lesssim _{n,k}(1+s)^{-\frac{n-1}{2}-k}.  \label{Bessel2}
\end{equation}
We first estimate $e^{-ir|x|}\bar{Z}(r|x|)$. Changing variable $r\rightarrow
r+r_{0},$ and letting $\Psi (v)=(r_{0}+r)^{n-1}\psi _{r_{0}}(r_{0}+r)$, we
get 
\begin{equation*}
\begin{split}
\tilde{I}_{1}& =\int_{0}^{\infty }e^{i\left( tp(r)-r|x|\right) }\psi (r)%
\overline{Z}(r|x|)r^{n-1}dr \\
& =\int_{-2\varepsilon }^{2\varepsilon }e^{i\left( tp(r)-r|x|\right) }\Psi
(r)\overline{Z}((r_{0}+r)|x|)dr
\end{split}%
\end{equation*}
and a first crude estimate allow us to conclude that 
\begin{equation}
|\tilde{I}_{1}|\lesssim _{n,k,\varepsilon }1  \label{EstimI1SmallTimes}
\end{equation}
which takes care of the small times $|t|\lesssim 1$. Thus, we now assume
that $t>1$. We consider the phase 
\begin{equation*}
\Omega (r,|x|,t)=\left(p(r)-r\frac{\vert x\vert}{t}\right) .
\end{equation*}
By \eqref{Estimp}, we directly compute that $p^{\prime }(r_{0})\neq 0$ and 
\begin{equation*}
p^{\prime \prime \prime }(r_{0})=\frac{4r_{0}^{4}-4r_{0}^{2}}{(1+r_{0}^{2})^{%
\frac{5}{2}}(2+r_{0}^{2})^{\frac{3}{2}}}\neq 0.
\end{equation*}

\textbf{Case 1} Suppose that $|x|\geq \frac{1}{4}p^{\prime }(r_{0})t$. Then,
since $\vert r\vert\le 2\varepsilon,$ 
\begin{equation*}
|\partial _{r}^{3}\Omega (r,|x|,t)|=|p^{\prime \prime \prime }(r)|>\frac{1}{%
2 }|p^{\prime \prime \prime }(r_{0})|,
\end{equation*}
if $\varepsilon>0$ is chosen sufficiently small, and using \eqref{Bessel2},
by the Van der Corput lemma (see e.g. Stein \cite{Ste}), we get that 
\begin{equation}
\begin{split}
|\tilde{I}_{1}|& \lesssim _{\varepsilon }|t|^{-\frac{1}{3}}\left( \sup_{k\in
\{0,1\},|r|\leq 2\varepsilon }|\overline{Z}((r_{0}+r)|x|)\partial ^{k}\Psi
(r)|+|x|\sup_{|r|\leq 2\varepsilon }|\Psi (r)\overline{Z}^{\prime
}((r_{0}+r)|x|)|\right) \\
& \lesssim _{\varepsilon }|t|^{-\frac{1}{3}}|x|^{-\frac{n-1}{2}} \\
& \lesssim _{\varepsilon }|t|^{-\frac{1}{3}-\frac{n-1}{2}}
\end{split}
\label{EstimI1LargeTimeCriticalPoint}
\end{equation}

\medskip

\textbf{Case 2} Suppose now that $|x|\leq \frac{1}{4}p^{\prime }(r_{0})t$.
Then 
\begin{equation*}
|\partial _{r}\Omega (r,|x|,t)|=|p^{\prime }(r_{0})|-\frac{|x|}{t}\geq
|p^{\prime }(r_{0})|/2
\end{equation*}
and therefore, using the nonstationary phase and the fact that $Z$ has all
derivatives bounded, we obtain that 
\begin{equation}
|\tilde{I}_{1}|\lesssim |t|^{-\frac{n}{2}}.
\label{EstimI1LargeTimeNonCriticalPoint}
\end{equation}

\medskip

The estimation of $e^{ir\vert x\vert}Z(r\vert x\vert)$ is easier. Proceeding
as above, we introduce 
\begin{equation*}
\tilde{I}_2 =\int_{-2\varepsilon }^{2\varepsilon }e^{i\left( tp(r)+r\vert
x\vert\right) }\Psi (r)Z((r_{0}+r)\vert x\vert)dr.
\end{equation*}
But the phase in $\tilde{I}_2$ satisfies 
\begin{equation*}
\vert \partial_r\Omega_2(r,\vert x\vert,t)\vert =\vert\partial_r\left(p(r)+r%
\frac{\vert x\vert}{t}\right)\vert\ge \vert p^\prime(r)\vert\gtrsim 1
\end{equation*}
and we can conclude as in Case 2 above to get 
\begin{equation*}
\vert\tilde{I}_2\vert\lesssim \vert t\vert^{-\frac{n}{2}}.
\end{equation*}
Now this, \eqref{EstimI1SmallTimes}, \eqref{EstimI1LargeTimeCriticalPoint}
and \eqref{EstimI1LargeTimeNonCriticalPoint} prove \eqref{LinearDecay1}.
\end{proof}

Now that we have dealt with the degeneracy at $r_{0}$, the other degeneracy
at $0$ and $\infty $ are more easily dealt with at the price of loosing
derivatives. To isolate these regions, we introduce two smooth cut-off
functions. We let $\psi _{0}$ and $\psi _{\infty }$ such that $0\leq \psi
_{0}+\psi _{\infty }\leq 1$, $\psi _{0}$ is supported on $%
(-r_{0}+\varepsilon ,r_{0}-\varepsilon )$, $\psi _{\infty }$ is supported on 
$\{| x |\geq r_{0}+\varepsilon \}$ and 
\begin{equation}
\psi _{0}+\psi _{r_{0}}+\psi _{\infty }=1.  \label{PartitionOfUnity}
\end{equation}
We first treat the case of small frequencies. We note that since $r_{0}$ is
the only positive root of $p^{\prime \prime }$, $p^{\prime \prime }(r)\neq 0$
for either $r\in (r_{0}+\varepsilon ,\infty )$ or $r\in $ $%
(0,r_{0}-\varepsilon ).$ Therefore we can apply Theorem 1 from Guo, Peng and
Wang \cite{GuoPenWan}, case $(a)$ and $(b)$ respectively, to obtain with %
\eqref{Estimp}:

\begin{lemma}
\label{LinearDecaySmallFreqLemma} There holds that, for all $f\in L^1$ 
\begin{equation}  \label{LinearDecaySmallFreq1}
\begin{split}
\Vert
e^{itp(\vert\nabla\vert)}\psi_0(\vert\nabla\vert)f\Vert_{L^\infty}%
\lesssim_{n,\varepsilon} (1+\vert t\vert)^{-\frac{n}{2}}\Vert
\vert\nabla\vert^{\frac{n-2}{2}}f\Vert_{L^1}.
\end{split}%
\end{equation}
\end{lemma}

\begin{lemma}
\label{LinearDecayHighFreqLemma} For all $f\in L^1$, there holds that 
\begin{equation}  \label{LinearDecayHighFreq}
\Vert
e^{itp(\vert\nabla\vert)}\psi_\infty(\vert\nabla\vert)f\Vert_{L^\infty}
\lesssim_{n,\varepsilon} \vert t\vert^{-\frac{n}{2}}\Vert \vert\nabla\vert^ 
\frac{n+2}{2} f\Vert_{B^0_{1,1}}.
\end{equation}
\end{lemma}

Finally, from Lemma \ref{LinearDecayMedFreqLemma} \ref%
{LinearDecaySmallFreqLemma} \ref{LinearDecayHighFreqLemma}, we can prove
Proposition \ref{PropLinDecay}.

\begin{proof}[Proof of Proposition \protect\ref{PropLinDecay}]
\eqref{LinDecay} follows directly from \eqref{LinearDecayMedFreqLemma}, %
\eqref{LinearDecaySmallFreq1} and \eqref{LinearDecayHighFreq}. In order to
get \eqref{L10Decay}, we interpolate between the isometric property 
\begin{equation*}
\Vert e^{itp(|\nabla |)}Pf\Vert _{L^{2}}=\Vert Pf\Vert _{L^{2}}
\end{equation*}%
for $P$ a Fourier projector and the various $L^{\infty }$ estimates.
Interpolating with \eqref{LinearDecayMedFreqLemma} gives that 
\begin{equation*}
\Vert e^{itp(|\nabla |)}\psi (|\nabla |)f\Vert _{L^{10}}\lesssim |t|^{-\frac{%
16}{15}}\Vert f\Vert _{L^{\frac{10}{9}}}.
\end{equation*}%
Interpolating with \eqref{LinearDecaySmallFreq1} gives 
\begin{equation*}
\Vert e^{itp(|\nabla |)}\psi _{0}(|\nabla |)f\Vert _{L^{10}}\lesssim |t|^{-%
\frac{6}{5}}\Vert f\Vert _{L^{\frac{10}{9}}}.
\end{equation*}%
Finally, interpolating with \eqref{LinearDecayHighFreq} and using the
inclusions of Besov spaces 
\begin{equation*}
L^{10}\subset B_{10,2}^{0}\hskip.2cm\hbox{and}\hskip.2cmB_{\frac{10}{9}%
,2}^{0}\subset L^{\frac{10}{9}},
\end{equation*}%
and Bernstein estimates \eqref{BernSobProp}, we get that 
\begin{equation*}
\begin{split}
\Vert e^{itp(|\nabla |)}\psi _{\infty }(|\nabla |)f\Vert _{L^{10}}^{2}&
\lesssim \sum_{N\geq 1}\Vert e^{itp(|\nabla |)}\psi _{\infty }(|\nabla
|)P_{N}f\Vert _{L^{10}}^{2} \\
& \lesssim |t|^{-\frac{12}{5}}\sum_{N\geq 1}N^{4}\Vert P_{N}f\Vert _{L^{%
\frac{10}{9}}}^{2} \\
& \lesssim |t|^{-\frac{12}{5}}\Vert f\Vert _{W^{2,\frac{10}{9}}}^{2}.
\end{split}%
\end{equation*}
Since for small time $t\leq 1$, we also have that 
\begin{equation*}
\Vert e^{itp(|\nabla |)}f\Vert _{L^{10}}\lesssim \Vert e^{itp(|\nabla
|)}f\Vert _{H^{\frac{6}{5}}}\lesssim \Vert f\Vert _{H^{\frac{6}{5}}}\lesssim
\Vert f\Vert _{W^{\frac{12}{5},\frac{10}{9}}}
\end{equation*}%
and since $f=\psi _{r_{0}}(|\nabla |)f+\psi _{0}(|\nabla |)f+\psi _{\infty
}(|\nabla |)f$, this ends the proof.
\end{proof}

%%%%%%%%%%%%%%%%%%%%%%%%%%%%%%%%%%%%%%%%%%%%%%%%%%%%%%%%%%%%%%%%%%%%%%%%%%%%%%%%%%%%%%%%%%%%%%%%%%%%%%%%%%%%%%%%%%%%%%%%%%%%%%%%%%%%%%%%%%%%%%%%%%%%%%%%%%%%%%%%%%%%%%%%%%%%%%%%%%%%%%%%%%%%%%%%%%%%%%%%%%%%%%%%%%%%%%%%%%%%%%%%%%%%%%%%%%%%%%%%%%%%%%%%%%%%%%%%%%%%%%%%%%%%%%%%%%%%%%%%%%%%%%%%%%%%%%%%%%%%%%%%%%%%%%%%%%%

\section{Normal form transformation}

\label{SecNot}

In this section, we derive the normal form transformation for $\alpha _{j}.$
Isolating linear, quadratic and higher order terms, we can rewrite the
Euler-Poisson system \eqref{EP} as follows: 
\begin{align}
& \partial _{t}\rho +\hbox{div}(v) & & +\hbox{div}(\rho v) & & =0
\label{EPrho} \\
& \partial _{t}v+\nabla \rho +\nabla \phi & & +(v\cdot \nabla )v-\nabla 
\frac{\rho ^{2}}{2} & & =-\nabla \left[ \ln (1+\rho )-\rho +\frac{\rho ^{2}}{
2}\right]  \label{EPv} \\
& \rho =(1-\Delta )\phi & & +\frac{\phi ^{2}}{2} & & +\left[ e^{\phi
}-1-\phi -\frac{\phi ^{2}}{2}\right] .  \label{EPphi}
\end{align}
The last line defines an operator $\rho \mapsto \phi (\rho )$ such that 
\begin{equation}  \label{DefOfElectricField}
\phi (\rho )=(1-\Delta )^{-1}\rho -\frac{1}{2}(1-\Delta )^{-1}\left[
(1-\Delta )^{-1}\rho \right] ^{2}+R(\rho )
\end{equation}
where $R$ satisfies good properties. We note that since $\nabla \times v=0$
there exists a function $\psi$ such that $v=\nabla \psi $ and consequently, $%
(v\cdot \nabla )v=\nabla \frac{|v|^{2}}{2}$. In terms of the velocity
potential $\psi $, we can rewrite the above system as

\begin{equation}  \label{EPtemp}
\begin{array}{c}
\partial _{t} 
\begin{pmatrix}
\rho \\ 
\psi%
\end{pmatrix}
+ 
\begin{pmatrix}
0 & \Delta \\ 
(1-\Delta )^{-1}+1 & 0%
\end{pmatrix}
\begin{pmatrix}
\rho \\ 
\psi%
\end{pmatrix}
\\ 
= 
\begin{pmatrix}
-\nabla \cdot (\rho \nabla\psi) \\ 
\frac{1}{2}(1-\Delta )^{-1}\left[ (1-\Delta )^{-1}\rho \right] ^{2}-R(\rho
)-\ln (1+\rho )+\rho -\frac{\rho ^{2}}{2}%
\end{pmatrix}
. \\ 
\end{array}%
\end{equation}

We denote the pair of eigenfunctions of the linear part as$%
\begin{pmatrix}
1 \\ 
\pm \frac{p(|\nabla |)}{i|\nabla |^{2}}%
\end{pmatrix}%
$ $=%
\begin{pmatrix}
1 \\ 
\pm \frac{q(|\nabla |)}{i|\nabla |}%
\end{pmatrix}%
,$ and recall $\alpha _{j}=\rho +\frac{(-1)^{j}i}{q(|\nabla |)}\mathcal{R }%
^{-1}v\equiv \rho +\frac{(-1)^{j}i}{q(|\nabla |)}|\nabla |\psi ,$ where $%
\mathcal{R}=\frac{\nabla}{\vert\nabla\vert}$ stands for the Riesz transform.
We can diagonalize the matrix as:

\begin{equation*}
\begin{pmatrix}
0 & \Delta \\ 
(1-\Delta )^{-1}+1 & 0%
\end{pmatrix}
= 
\begin{pmatrix}
1 & 1 \\ 
\frac{q(|\nabla |)}{i|\nabla |} & -\frac{q(|\nabla |)}{i|\nabla |}%
\end{pmatrix}
\begin{pmatrix}
ip(|\nabla |) & 0 \\ 
0 & -ip(|\nabla |)%
\end{pmatrix}
\begin{pmatrix}
\frac{1}{2} & \frac{i|\nabla |}{2q(|\nabla |)} \\ 
\frac{1}{2} & -\frac{i|\nabla |}{2q(|\nabla |)}%
\end{pmatrix}%
\end{equation*}
Now, with $\alpha $ given in \eqref{DefOfAlpha}, using that $\mathcal{R}%
^{-1} \nabla=\vert\nabla\vert\mathcal{R}^{-1}\frac{\nabla }{|\nabla |}%
=|\nabla |,$ and

\begin{equation*}
\hbox{div}(v)=\hbox{div}(\frac{\nabla }{|\nabla |}\mathcal{R}
^{-1}v)=-|\nabla |\mathcal{R}^{-1}v,
\end{equation*}
we diagonalize the matrix and rewrite \eqref{EPtemp} in terms of $\alpha $ as

\begin{equation}  \label{EPShort}
(\partial _{t}+(-1)^{j}ip(|\nabla |))\alpha _{j}=Q_{j}(\alpha )+\mathcal{N}
_{j}
\end{equation}
where $Q_{2}=\bar{Q}_{1},$ and $\mathcal{N}_{2}=\mathcal{\bar{N}}_{1}$ such
that the quadratic term $Q_{1}$ and the cubic term $\mathcal{N}_{1}$ take
the form: 
\begin{equation}  \label{DefOfN}
\begin{split}
&Q_{j}=-\text{div}(\rho v)+(-1)^{j}\frac{i|\nabla |}{2q(|\nabla |)}\left\{
(1-\Delta )^{-1}[(1-\Delta )^{-1}\rho ]^{2}-\rho ^{2}-|v|^{2}\right\} \\
&\mathcal{N}_{j}=(-1)^{j}\frac{i|\nabla |}{q(|\nabla |)}\left[ \ln (1+\rho
)-\rho +\frac{\rho ^{2}}{2}-R(\rho )\right] .
\end{split}%
\end{equation}
The most important step is to study the linear profiles 
\begin{equation*}
\omega _{j}(t)=e^{(-1)^{j}itp(|\nabla |)}\alpha _{j}(t),
\end{equation*}
so that its temporal derivatives are of at least of quadratic order: 
\begin{equation}  \label{OmegaDerivatives}
\partial _{t}\omega _{j}=e^{(-1)^{j}ip(|\nabla |)t}\{Q_{j}(\alpha )+ 
\mathcal{N}_{j}\}.
\end{equation}
Plugging $\rho =\frac{\alpha _{1}+\alpha _{2}}{2}$ and $v=\frac{\nabla
p(|\nabla |)}{-\Delta }\frac{\alpha_{1}-\alpha_{2}}{2i}$ into $Q_{j},$ we
now compute the Fourier transform of $Q_{j}(\alpha )$ as 
\begin{equation}  \label{EquationForQ}
\begin{split}
& \hat{Q}_{j}(\alpha )(t,\xi ) \\
& =\int_{\mathbb{R}^{3}}\Big[-\frac{1}{4}\frac{\xi \cdot \eta }{|\eta |}
q(\eta )\hat{\alpha}_{1}(\xi -\eta )\hat{\alpha}_{1}(\eta )+\frac{1}{4} 
\frac{\xi \cdot \eta }{|\eta |}q(\eta )\hat{\alpha}_{2}(\xi -\eta )\hat{
\alpha}_{2}(\eta ) \\
& +\frac{i|\xi |}{8q(\xi )}\left(\frac{(\xi -\eta )\cdot \eta }{|\xi -\eta
||\eta |}q(\xi -\eta )q(\eta )-1+\frac{1}{\langle \xi \rangle ^{2}\langle
\xi -\eta \rangle ^{2}\langle \eta \rangle ^{2}}\right) \hat{\alpha}_{1}(\xi
-\eta )\hat{\alpha}_{2}(\eta ) \\
& -\frac{(-1)^{j}i|\xi |}{8q(\xi )}\left(\frac{(\xi -\eta )\cdot \eta }{
|\xi -\eta ||\eta |}q(\xi -\eta )q(\eta )-1+\frac{1}{\langle \xi \rangle
^{2}\langle \xi -\eta \rangle ^{2}\langle \eta \rangle ^{2}}\right) \hat{
\alpha}_{1}(\xi -\eta )\hat{\alpha}_{1}(\eta ) \\
& +\frac{1}{4}\frac{\xi \cdot \eta }{|\eta |}q(\eta )\hat{\alpha}_{1}(\xi
-\eta )\hat{\alpha}_{2}(\eta ) -\frac{1}{4}\frac{\xi \cdot \eta }{|\eta |}%
q(\eta )\hat{\alpha}_{2}(\xi -\eta )\hat{\alpha}_{1}(\eta ) \\
& -\frac{(-1)^{j}i|\xi |}{8q(\xi )}\left(\frac{(\xi -\eta )\cdot \eta }{
|\xi -\eta ||\eta |}q(\xi -\eta )q(\eta )-1+\frac{1}{\langle \xi \rangle
^{2}\langle \xi -\eta \rangle ^{2}\langle \eta \rangle ^{2}}\right) \hat{
\alpha}_{2}(\xi -\eta )\hat{\alpha}_{2}(\eta )\Big]d\eta \\
& \equiv \int_{\mathbb{R}^{3}}\big[m_{rl}^{j}(\xi,\eta)\hat{\alpha}_{r}(\xi
-\eta ) \hat{\alpha}_{l}(\eta )\big](s)ds.
\end{split}%
\end{equation}

We now integrate \eqref{EPShort} to get 
\begin{equation}  \label{DuhamelForAlpha}
\begin{split}
\hat{\alpha}_{j}(t) &=e^{(-1)^{j+1}ip(|\xi |)t}\hat{\alpha}
_{j}(0)+\int_{0}^{t}e^{(-1)^{j+1}ip(|\xi |)(t-s)}\hat{Q}_{j}(\alpha )(s)ds \\
&+\int_{0}^{t}e^{(-1)^{j+1}ip(|\xi |)(t-s)}\mathcal{\hat{N}}_{1}(\alpha
)(s)ds \\
&=e^{(-1)^{j+1}ip(|\xi |)t}\hat{\alpha}_{j}(0)+\int_{0}^{t}e^{(-1)^{j+1}ip(|%
\xi |)(t-s)}\mathcal{\hat{N}}_{j}(\alpha )(s)ds \\
&+e^{(-1)^{j+1}ip(|\xi |)t}\int_{0}^{t}e^{(-1)^{j+1}ip(|\xi |)s}m_{rl}^{j}%
\hat{\alpha}_{r}(\xi -\eta )\hat{\alpha}_{l}(\eta )\big](s)ds(s)ds.
\end{split}%
\end{equation}
The crucial step is to replace $\hat{\alpha}_{j}(s)=e^{(-1)^{j+1}ip(|\xi
|)s} \hat{\omega}_{j}(s)$ in the third term, which then takes the form 
\begin{equation}  \label{EquaPsi}
\hat{\Psi}_j(\alpha)=e^{(-1)^{j+1}ip(|\xi
|)t}\sum_{r,l=1}^{2}\int_{0}^{t}\int_{\mathbb{R} ^{3}}m_{rl}^{j}e^{is\Phi
_{rl}}\hat{\omega}_{r}(\xi -\eta )\hat{\omega} _{l}(\eta )d\eta ds.
\end{equation}
Here $\hat{\omega}_{1}(\xi)=e^{itp(\xi )}\hat{\alpha}_{1}(\xi )$ and $\hat{%
\omega}_{2}(\xi )=e^{-itp(\xi )}\hat{\alpha}_{2}(\xi)=\overline{\hat{\omega}}%
_{1}(\xi ),$ 
\begin{equation}  \label{DefOfPhim}
\begin{split}
\Phi _{rl}(\xi ,\eta )& =(-1)^{j+1}p(\xi )+(-1)^{r+1}p(\xi -\eta
)+(-1)^{l+1}p(\eta ),\hskip.1cm\hbox{and} \\
m_{rl}^{j}(\xi ,\eta )& =|\xi |n_{1rl}^{j}(\xi )n_{2rl}^{j}(\xi -\eta
)n_{3rl}^{j}(\eta ),
\end{split}%
\end{equation}
is a factorable multiplier defined in \eqref{EquationForQ}, where the $%
n_{olk}^{j}$ are either smooth functions or product of a smooth function
with the angle function $x\mapsto \frac{x}{|x|}$. More specifically, there
are only four different phases of the following: 
\begin{equation}  \label{phase}
\begin{split}
\Phi _{1}(\xi ,\xi -\eta ,\eta )& =p(\xi )-p(\xi -\eta )-p(\eta ) \\
\Phi _{2}(\xi ,\xi -\eta ,\eta )& =p(\xi )+p(\xi -\eta )+p(\eta ) \\
\Phi _{3}(\xi ,\xi -\eta ,\eta )& =p(\xi )-p(\xi -\eta )+p(\eta ) \\
\Phi _{4}(\xi ,\xi -\eta ,\eta )& =p(\xi )+p(\xi -\eta )-p(\eta ). \\
&
\end{split}%
\end{equation}
Integrating by parts in $s$ in the integral in $\Psi $, and making use of
the fact that $\partial _{t}\hat{\omega}$ is at least quadratic by %
\eqref{OmegaDerivatives}, we obtain from \eqref{EquaPsi} that\footnote{%
for notational simplicity, we do not distinguish $m^j_{lr}(\xi,\eta)$ and $%
m^j_{rl}(\xi,\xi-\eta)$.} 
\begin{equation*}
\begin{split}
& e^{(-1)^{j}ip(|\xi |)t}\hat{\Psi}_j(\alpha)(t,\xi ) \\
& =\sum_{r,l=1}^{2}\left[ \int_{\mathbb{R}^{3}}\frac{m_{rl}^{j}}{i\Phi _{rl}}
e^{is\Phi _{rl}}\hat{\omega}_{r}(\xi -\eta )\hat{\omega}_{l}(\eta )d\eta %
\right] _{s=0}^{t} \\
& +2\sum_{r,l=1}^{2}\int_{0}^{t}\int_{\mathbb{R}^{3}}i\frac{%
m_{rl}^{j}(\xi,\eta)}{\Phi _{rl}}e^{is\Phi _{rl}}\hat{\omega}_{r}(\xi -\eta
)\partial _{t}\hat{\omega} _{l}(\eta )d\eta ds \\
& =i\sum_{r,l=1}^{2}\int_{\mathbb{R}^{3}}\frac{m_{rl}^{j}}{\Phi _{rl}}\hat{
\omega}_{r}(0,\xi -\eta )\hat{\omega}_{l}(0,\eta )d\eta \\
& +\sum_{r,l=1}^{2}e^{(-1)^jitp(\xi )}\int_{\mathbb{R}^{3}}\frac{m_{rl}^{j}}{
i\Phi _{rl}}\hat{\alpha}_{r}(t,\xi -\eta )\hat{\alpha}_{l}(t,\eta )d\eta \\
& +2\sum_{r,l,r_{1},l_{1}=1}^{2}\int_{0}^{t}\int_{\mathbb{R}^{3}}\frac{
im_{rl}^{j}(\xi ,\eta )m_{r_{1}l_{1}}^{l}(\eta ,\zeta )}{\Phi _{rl}}
e^{is\Phi _{rl}}\hat{\omega}_{r}(\xi -\eta )e^{is(-1)^{l}p(|\eta |)} \hat{Q}%
_j(\alpha)(\eta)d\eta ds \\
& +2\sum_{r,l=1}^{2}\int_{0}^{t}\int_{\mathbb{R}^{3}}\frac{im_{rl}^{j}}{\Phi
_{rl}}e^{is\Phi _{rl}}\hat{\omega}_{r}(\xi -\eta )e^{is(-1)^{l}p(\eta )} 
\hat{\mathcal{N}_{l}}(\eta )d\eta ds
\end{split}%
\end{equation*}
We then change back to $\hat{\omega}_{j}(s)=e^{(-1)^{r+1}ip(|\xi |)s}\hat{
\alpha}_{j}(s),$ and using \eqref{DefOfN}, we write

\begin{eqnarray}
&&e^{(-1)^{j}ip(|\xi |)t}\left(\hat{\alpha}_{j}(t)+\mathfrak{B}_{j}(\alpha )
\right)  \notag \\
&=&\hat{\alpha}_{j}(0)+\mathfrak{B}_{j}(\alpha
(0))+\int_{0}^{t}e^{(-1)^{j}ip(|\xi |)s}\hat{Q}_{j}(\alpha
)(s)ds+\int_{0}^{t}e^{(-1)^{j}ip(|\xi |)s}\mathcal{\hat{N}}_{j}(\alpha )(s)ds
\notag \\
&=&\hat{\alpha}_{j}(0)+\mathfrak{B}_{j}(\alpha
(0))+\int_{0}^{t}e^{(-1)^{j}ip(|\xi |)s}\mathcal{\hat{N}}_{j}(\alpha )(s)ds 
\notag \\
&&+\int_{0}^{t}e^{(-1)^{j}ip(|\xi |)s}\frac{im_{lk}^{j}(\xi ,\eta )}{\Phi
_{lk}}\hat{\alpha}_{r}(\xi -\eta )\hat{h}_{l}(\alpha (\eta ))(s)ds(s)ds
\label{alpha}
\end{eqnarray}
where the normal form transformation is 
\begin{equation}
\mathcal{F}\mathfrak{B}_j(\alpha _{j})(\xi )=\sum_{r,l=1}^{2}\int_{\mathbb{R}
^{3}}\frac{m_{rl}^{j}}{i\Phi _{rl}}\hat{\alpha}_{r}(\xi -\eta )\hat{\alpha}
_{l}(\eta )d\eta  \label{DefOfB}
\end{equation}
and 
\begin{equation}  \label{DefOfH}
\hat{h}_{l}(\alpha )\equiv \int_{\mathbb{R}^{3}}m_{r_{1}l_{1}}(\eta ,\zeta ) 
\hat{\alpha}_{r_{1}}(\eta -\zeta )\hat{\alpha}_{l_{1}}(\zeta )d\zeta + 
\mathcal{\hat{N}}_{l}
\end{equation}
is the associated cubic nonlinearity.

\medskip

We next show that $h(\alpha )$ behaves like a quadratic term in $\alpha .$

\begin{lemma}
Assuming that $\alpha$ has small $X$-norm, then 
\begin{equation}  \label{PropertiesOfN}
\Vert \vert\nabla\vert^{-1}h(\alpha (t))\Vert_{H^{2k}}+\Vert
\vert\nabla\vert^{-1}\mathcal{N}\Vert_{H^{2k}}\lesssim (1+t)^{-\frac{16}{15}%
}\Vert \alpha \Vert_{X}^{2}.
\end{equation}
\end{lemma}

\begin{proof}
When $h$ is a product of $\alpha $'s, this follows directly from the Sobolev
embedding $L^{\infty }\subset W^{\frac{3}{10},10}$.

\medskip

When $h=\mathcal{N}$, we see from \eqref{DefOfN} that, except for the term
involving $R$, a similar proof works. For the terms involving $R$, we
proceed as follows: Letting $E(x)=e^{x}-1-x-\frac{x^{2}}{2}$, we see from %
\eqref{EPphi}. \eqref{DefOfElectricField} that 
\begin{equation}  \label{EquForR}
\begin{split}
& (1-\Delta )R+\frac{1}{2}\left[ (1-\Delta )^{-1}\rho -(1-\Delta
)^{-1}\left( (1-\Delta )^{-1}\rho \right) ^{2}\right] R+\frac{R^{2}}{2} \\
& +E((1-\Delta )^{-1}\rho -\frac{1}{2}(1-\Delta )^{-1}\left[ (1-\Delta
)^{-1}\rho \right] ^{2}+R) \\
& =\frac{1}{2}(1-\Delta )^{-1}\rho \left[ (1-\Delta )^{-1}\left[ (1-\Delta
)^{-1}\rho \right] ^{2}\right] -\frac{1}{8}\left[ (1-\Delta )^{-1}\left[
(1-\Delta )^{-1}\rho \right] ^{2}\right] ^{2}.
\end{split}%
\end{equation}

In order to solve \eqref{EquForR}, we define the following iterative scheme.
For $\rho $ sufficiently small in $X$-norm, we let 
\begin{equation*}
\begin{split}
& R_{0}=0 \\
& (1-\Delta )R_{k+1}=-\frac{1}{2}\left[ (1-\Delta )^{-1}\rho -(1-\Delta
)^{-1} \left( (1-\Delta )^{-1}\rho \right) ^{2}\right] R_{k}-\frac{R_{k}^{2}%
}{2} \\
& -E((1-\Delta )^{-1}\rho -\frac{1}{2}(1-\Delta )^{-1}\left[ (1-\Delta
)^{-1}\rho \right] ^{2}+R_{k}) \\
& +\frac{1}{2}(1-\Delta )^{-1}\rho \left[ (1-\Delta )^{-1}\left[ (1-\Delta
)^{-1}\rho \right] ^{2}\right] -\frac{1}{8}\left[ (1-\Delta )^{-1}\left[
(1-\Delta )^{-1}\rho \right] ^{2}\right] ^{2}
\end{split}%
\end{equation*}
We see that, if $s>3/2$ and $p\ge 2$, using the tame estimate \eqref{tameEst}
\begin{equation}  \label{EstimForNormRk}
\begin{split}
\Vert R_{k+1}\Vert _{W^{s+2,p}}& \lesssim \Vert \rho \Vert _{L^\infty}\Vert
R_{k}\Vert _{W^{s,p}}+\Vert \rho\Vert_{W^{s-2,p}}\Vert
R_k\Vert_{L^\infty}+\Vert R_{k}\Vert _{L^\infty}\Vert R_k\Vert_{W^{s,p}} \\
&+\Vert \rho \Vert_{L^\infty}\left(\Vert\rho\Vert _{W^{s-2,p}}^{2}+\Vert
\rho \Vert _{W^{s-2,p}}^{3}\right) \\
&+C\left( \Vert \rho \Vert _{L^{\infty }}+\Vert R_{k}\Vert _{L^{\infty
}}\right) \left( \Vert R_{k}\Vert _{W^{s,p}}+\Vert \rho \Vert
_{W^{s-2,p}}\right) ^{3}
\end{split}%
\end{equation}
and, assuming that 
\begin{equation*}
\sup_{k}\Vert R_{k}\Vert _{L^{\infty }}+\Vert \rho \Vert _{L^{\infty }}\leq 2
\end{equation*}
we also see that 
\begin{equation*}
\Vert R_{k+1}-R_{k}\Vert _{H^{2}}\lesssim \left( \Vert \rho \Vert
_{L^{\infty }}+\sup_{k}\Vert R_{k}\Vert _{L^{\infty }}\right) \Vert
R_{k}-R_{k-1}\Vert _{L^{2}}.
\end{equation*}%
Hence, if $\Vert \rho \Vert _{X}<1$ is sufficiently small, there holds that 
\begin{equation*}
(1+t)^{\frac{16}{15}}\Vert R_{k}\Vert _{W^{s+2,10}}+\Vert R_{k}\Vert
_{H^{2(s+1)}}\lesssim \Vert \rho \Vert _{X}^{3}\lesssim 1
\end{equation*}
for all $0\le s\le k$ and that $(R_{k})_{k}$ is a Cauchy sequence in $H^{2}$%
, hence converges to a unique limit $R=R(\rho )$, the given function which
solves \eqref{EquForR} and satisfies 
\begin{equation}  \label{ControlOfR}
(1+t)^{\frac{16}{15}}\Vert R(\rho )\Vert _{W^{k+2,10}}+\Vert R(\rho )\Vert
_{H^{2(k+1)}}\lesssim \Vert \alpha \Vert _{X}^{3}.
\end{equation}
Using now that $W^{2,10}\subset L^\infty$, one recovers from %
\eqref{EstimForNormRk} that for all $k$, 
\begin{equation*}
(1+t)^\frac{16}{15}\Vert R_k(t)\Vert_{H^{2k+2}}\lesssim \Vert
\alpha\Vert_X^3.
\end{equation*}
Passing to the limit in $k$, we finish the proof of \eqref{PropertiesOfN}.
\end{proof}

%%%%%%%%%%%%%%%%%%%%%%%%%%%%%%%%%%%%%%%%%%%%%%%%%%%%%%%%%%%%%%%%%%%%%%%%%%%%%%%%%%%%%%%%%%%%%%%%%%%%%%%%%%%%%%%%%%%%%%%%%%%%%%%%%%%%%%%%%%%%%%%%%%%%%%%%%%%%%%%%%%%%%%%%%%%%%%%%%%%%%%%%%%%%%%%%%%%%%%%%%%%%%%%%%%%%%%%%%%%%%%%%%%%%%%%%%%%%%%%%%%%%%%%%%%%%%%%%%%%%%%%%%%%%%%%%%%%%%%%%%%%%%%%%%%%%%%%%%%%%%%%%%%%%%%%%%%%%%%%%%%%%%%%%%%%%%%%%%%%%%%%%%%%%%%%%%%%%%%%%%%%%%%%%%%%%%%%%%%%%%%%%%%%%%%%%%%%%%%%%%%%%%%%%%%%%%%%%%%%%%%%%%%%%%%%%%%%%

\section{The $L^2$-type norm}

\label{SecH-1}

In this section, we get control on the first part of the $X$-norm, namely,
we control the $L^2$-based norms as follows

\begin{proposition}
\label{ControlL2NormProp} Let $\alpha$ correspond to a solution of \eqref{EP}
by \eqref{DefOfAlpha}, then if $\alpha$ has small $X$-norm there holds that 
\begin{equation*}
\Vert \alpha\Vert_{H^{-1}\cap H^{2k}}\lesssim \Vert
\alpha(0)\Vert_{Y}+\Vert\alpha\Vert_X^\frac{3}{2}.
\end{equation*}
\end{proposition}

The remaining of this section is devoted to the proof of Proposition \ref%
{ControlL2NormProp}. We first control the high derivatives and then the $%
H^{-1}$-norm.

\subsection{The Energy estimate}

\label{SecEnergyEstimate}

In this subsection, we use energy methods to control high derivatives of the
solution in $L^2$, assuming a control on the $X$-norm, and most notably
integrability of the solution in $L^{10}$-norms.

In order to prove this, we rewrite \eqref{EP} into the symmetrized form 
\begin{equation}
\partial _{t}u+A_{j}(u)\partial _{j}u=(0,-\nabla \phi )  \label{SymHypSys}
\end{equation}
where $u=(\ln (1+\rho ),v_{1},v_{2},v_{3})$, 
\begin{equation*}
A_{j}= 
\begin{pmatrix}
v_{j} & e_{j}^{T} \\ 
e_{j} & v_{j}I_{3}%
\end{pmatrix}
.
\end{equation*}
Now, for a multi-index $\tau $, we derive \eqref{SymHypSys} $\tau $ times
and take the scalar product with $D^{\tau }u$ to get 
\begin{equation*}
\begin{split}
\frac{1}{2}\frac{d}{dt}\Vert D^{\tau }u\Vert _{L^{2}}^{2}&
=-(A_{j}(u)D^{\tau }\partial _{j}u,D^{\tau }u)_{L^{2}\times L^{2}} \\
& -\sum_{\gamma <\tau }c_{\gamma }(D^{\tau -\gamma }[A_{j}(u)]D^{\gamma
}\partial _{j}(u),D^{\tau }u)-(\nabla D^{\tau }\phi ,D^{\tau
}v)_{L^{2}\times L^{2}} \\
& \lesssim \Vert \hbox{div}(v)\Vert_{L^\infty}\Vert D^{\tau }u\Vert
_{L^{2}}^{2}+\sum_{\gamma <\tau }c_{\gamma }\vert (D^{\tau -\gamma
}[A_{j}(u)]D^{\gamma }\partial _{j}(u),D^{\tau }u)_{L^{2}\times L^{2}}\vert
\\
& +\vert (D^{\tau }\phi ,D^{\tau }\hbox{div}(v))_{L^{2}\times L^{2}}\vert. \\
&
\end{split}%
\end{equation*}
Besides, using \eqref{EPrho} and \eqref{DefOfElectricField}, one sees that 
\begin{equation*}
\begin{split}
(D^{\tau }\phi ,D^{\tau }\hbox{div}(v))& =(D^{\tau }(1-\Delta )^{-1}\rho
,D^{\tau }\hbox{div}(v))-(\nabla D^{\tau }\tilde{R}(\rho ),D^{\tau }v) \\
& =-(D^{\tau }(1-\Delta )^{-1}\rho ,D^{\tau }\partial _{t}\rho )-(D^{\tau
}(1-\Delta )^{-1}\rho ,D^{\tau }\hbox{div}(\rho v)) \\
& -(\nabla D^{\tau }\tilde{R}(\rho ),D^{\tau }v)
\end{split}%
\end{equation*}
with 
\begin{equation*}
\tilde{R}(\rho )=\frac{1}{2}(1-\Delta )^{-1}\left[ (1-\Delta )^{-1}\rho %
\right] ^{2}-R(\rho )
\end{equation*}
and $R$ given in \eqref{DefOfElectricField}. Now, using Lemma \ref%
{LemProdEnergy}, we remark that for all $\gamma<\tau$, there holds that 
\begin{equation*}
\Vert D^{\tau-\gamma}uD^\gamma\partial_ju\Vert_{L^2}\lesssim \Vert
u\Vert_{W^{2,10}}\Vert u\Vert_{H^{\vert\tau\vert}}
\end{equation*}
and combining this with \eqref{ControlOfR}, we obtain 
\begin{equation*}
\begin{split}
\frac{1}{2}\frac{d}{dt}\left( \Vert D^{\tau }u\Vert _{L^{2}}^{2}+\Vert
(1-\Delta )^{-\frac{1}{2}}D^{\tau }\rho \Vert _{L^{2}}^{2}\right) & \lesssim
\Vert u\Vert _{W^{2,10}}^2\Vert u\Vert _{H^{\vert\tau\vert}}+\Vert
R\Vert_{H^{\vert\tau\vert+1}}\Vert u\Vert_{H^{\vert\tau\vert}} \\
&\lesssim (1+t)^{-\frac{16}{15}}\Vert u\Vert_X^3
\end{split}%
\end{equation*}
as long as $\tau \le 2k$ and that $\Vert \alpha\Vert _{X}$ is sufficiently
small. Finally, integrating this in time and remarking that 
\begin{equation*}
\rho =\hbox{Re}(\alpha )\hskip.2cm\hbox{and}\hskip.2cmv=q(|\nabla |)\mathcal{%
R}\hbox{Im}(\alpha ),
\end{equation*}
we obtain that 
\begin{equation}
\Vert u\Vert _{H^{\tau }}^{2}\lesssim \Vert u(0)\Vert _{H^{\tau }}^{2}+\Vert
u\Vert _{X}^{3}  \label{EnergyEstimate}
\end{equation}
provided that $\tau \leq 2k$. Since control of $\ln (1+\rho )$ in $%
L_{t}^{\infty }H_{x}^{\tau }$-norm gives control of $\rho $ in $%
L_{t}^{\infty }H_{x}^{\tau }$ -norm, this gives us the global bound on the
derivatives we needed.

\subsection{The $H^{-1}$-norm}

In this section, we control the $H^{-1}$ norm of the solution, which the
other $L^2$ component of the component of the $X$-norm. We use %
\eqref{DuhamelForAlpha} and we first deal with the quadratic terms $%
Q_j(\alpha )$, whose contribution can be written as a finite sum of terms
like (recall that $\alpha_1=\overline{\alpha}_2$) 
\begin{equation*}
I=\mathcal{F}^{-1}\int_{0}^{t}e^{i(t-s)p(\xi )}|\xi |\frac{m}{|\xi |}\hat{
\alpha}(\xi -\eta )\hat{\alpha}(\eta )d\eta
\end{equation*}
where, from \eqref{DefOfPhim} we see that one can write 
\begin{equation*}
m=|\xi |n_{1}(\xi )n_{2}(\xi -\eta )n_{3}(\eta )
\end{equation*}
with $n_{i}(\zeta )=\frac{\zeta }{|\zeta |}\tilde{n}(\zeta )$ or $%
n_{i}(\zeta )=\tilde{n}(\zeta )$ for $\tilde{n}$ an $S^{0}$-symbol. In
particular, 
\begin{equation*}  \label{SymbolEst}
\Vert n_{i}(|\nabla |)f\Vert _{L^{r}}\lesssim \Vert f\Vert _{L^{r}}
\end{equation*}
for $1<r<\infty $. We use a standard energy estimate and the inclusion $%
L^\infty\subset W^{1,10}$ to get 
\begin{equation*}
\begin{split}
\Vert \frac{I}{|\xi |}\Vert _{L^{2}}& \lesssim \int_{0}^{t}\Vert \int_{%
\mathbb{R}^{3}}\frac{m}{|\xi |}\hat{\alpha}(\xi -\eta )\hat{\alpha}(\eta
)d\eta ds\Vert _{L^{2}}ds \\
& \lesssim \int_{0}^{t}\Vert \int_{\mathbb{R}^{3}}\left( n_{2}(\xi -\eta )%
\hat{\alpha}(\xi -\eta )\right) \left( n_{3}(\eta )\hat{\alpha}(\eta
)\right) d\eta \Vert _{L^{2}}ds \\
& \lesssim \int_{0}^{t}\Vert \left( n_{2}(|\nabla |)\alpha \right) \left(
n_{3}(|\nabla |)\alpha \right) \Vert _{L^{2}}ds \\
& \lesssim \int_{0}^{t}\Vert n_{2}(|\nabla |)\alpha \Vert _{L_{t}^{\infty
}L_{x}^{2}}\Vert n_{3}(|\nabla |)\alpha (s)\Vert _{L_{x}^{\infty }}ds \\
& \lesssim \Vert \alpha \Vert _{X}\int_{0}^{t}\Vert (1-\Delta )^{\frac{1}{2}%
}n_{3}(|\nabla |)\alpha (s)\Vert _{L_{x}^{10}}ds \\
& \lesssim \Vert \alpha \Vert _{X}\int_{0}^{t}\Vert (1-\Delta )^{\frac{1}{2}%
}\alpha (s)\Vert _{L_{x}^{10}} \lesssim \Vert \alpha \Vert
_{X}^{2}\int_{0}^{t}\frac{ds}{(1+s)^{\frac{16}{ 15}}} \\
&\lesssim \Vert \alpha \Vert _{X}^{2}.
\end{split}%
\end{equation*}
Next we control the contribution of the cubic term $\mathcal{N}$ as follows
using the fact that $e^{itp(\vert\nabla\vert)}$ is a unitary operator and %
\eqref{PropertiesOfN}, 
\begin{equation*}
\begin{split}
\Vert |\nabla |^{-1}\int_{0}^{t}e^{-i(t-s)p(|\nabla |)}\mathcal{N}(s)ds\Vert
_{L^{2}}& \lesssim \int_{0}^{t}\Vert |\nabla |^{-1}\mathcal{N}\Vert
_{L^{2}}ds \\
& \lesssim \Vert \alpha \Vert _{X}^{2}\int_{0}^{t} \frac{ds}{(1+s)^{\frac{16%
}{15}}}\lesssim \Vert \alpha \Vert _{X}^{2}.
\end{split}%
\end{equation*}
Combining the two above estimates give that 
\begin{equation}
\begin{split}
\Vert |\nabla |^{-1}\alpha\Vert _{L_{t}^{\infty }L_{x}^{2}}& \lesssim \Vert
|\nabla |^{-1}\alpha _{0}\Vert _{L^{2}}+\Vert \frac{I(\xi )}{ |\xi |}\Vert
_{L^{2}}+\Vert \int_{0}^{t}\frac{e^{-i(t-s)p(|\xi |)}}{|\xi |} \hat{\mathcal{%
N}}(s)ds\Vert _{L^{2}} \\
& \lesssim \Vert \alpha _{0}\Vert _{Y}+\Vert \alpha \Vert _{X}^{2}
\end{split}
\label{H1Norm}
\end{equation}
so that we control the first part in the $X$-norm.

%%%%%%%%%%%%%%%%%%%%%%%%%%%%%%%%%%%%%%%%%%%%%%%%%%%%%%%%%%%%%%%%%%%%%%%%%%%%%%%%%%%%%%%%%%%%%%%%%%%%%%%%%%%%%%%%%%%%%%%%%%%%%%%%%%%%%%%%%%%%%%%%%%%%%%%%%%%%%%%%%%%%%%%%%%%%%%%%%%%%%%%%%%%%%%

\section{Bilinear Multiplier Theorem}

\label{SecMult}

\subsection{A general multiplier theorem}

In order to control the last part of the norm, we need to deal with bilinear
terms in \eqref{DefOfB}, \eqref{DefOfH} which involve convolution with a
singular symbol. Note that since $p(0)=0$, the symbol is quite singular on
the whole parameter space and especially near $(\xi ,\eta )=(0,0)$. In
particular, we cannot use the traditional Coifman-Meyer multiplier theorem 
\cite{CoiMey}, or a more refined version as in Muscalu, Pipher, Tao and
Thiele \cite{Mus,MusPipTaoThi} since in all these cases, the multiplier need
to satisfy some homogeneity conditions. In order to overcome this we use
estimates inspired from Gustafson, Nakanishi and Tsai \cite{GNT} that we
present now. Although most of the results in this subsection are essentially
contained in Gustafson, Nakanishi and Tsai \cite{GNT}, for selfcontainness,
we give a direct proof.

We introduce the following multiplier norm: 
\begin{equation}
\Vert \mathfrak{m}\Vert _{M_{\xi ,\eta }^{s,b}}=\sum_{N\in 2^{\mathbb{Z}%
}}\Vert P_{N}^{\eta }\mathfrak{m}(\xi ,\eta )\Vert _{L_{\xi }^{b}\dot{H}%
_{\eta }^{s}}  \label{MsNorm}
\end{equation}%
and we let $\mathcal{M}_{\xi ,\eta }^{s}=\mathcal{M}_{\xi ,\eta }^{s,\infty
} $, which will be the norm that we mostly use. To a multiplier $\mathfrak{m}
$, we associate the bilinear \textquotedblleft
pseudo-product\textquotedblright\ operator 
\begin{equation}
B[f,g]=\mathcal{F}_{\xi }^{-1}\int_{\mathbb{R}^{3}}\mathfrak{m}(\xi ,\eta )%
\hat{f}(\xi -\eta )\hat{g}(\eta )d\eta .  \label{b}
\end{equation}
Our goal in this section is to obtain robust estimates on $B$.

\begin{lemma}
\label{infinity}If $\Vert \mathfrak{m}\Vert _{L_{\xi }^{\infty }\dot{H}
_{\eta }^{s-\varepsilon }}+\Vert \mathfrak{m}\Vert _{L_{\xi }^{\infty }\dot{%
H
}_{\eta }^{s+\varepsilon }}<\infty $, then the $M_{\xi ,\eta }^{s}$-norm
of $\mathfrak{m}$ is finite.
\end{lemma}

\begin{proof}
Indeed, we have that 
\begin{equation*}
\begin{split}
\Vert P_{N}^{\eta }\mathfrak{m}(\xi ,\eta )\Vert _{L_{\xi }^{\infty }\dot{H}
_{\eta }^{s}}& \leq \min (N^{-\varepsilon }\Vert \mathfrak{m}\Vert _{L_{\xi
}^{\infty }\dot{H}_{\eta }^{s+\varepsilon }},N^{\varepsilon }\Vert \mathfrak{%
m}\Vert _{L_{\xi }^{\infty }\dot{H}_{\eta }^{s-\varepsilon }}),\hskip.2cm %
\hbox{so that} \\
\sum_{N}\Vert P_{N}^{\eta }\mathfrak{m}(\xi ,\eta )\Vert _{L_{\xi }^{\infty
} \dot{H}_{\eta }^{s}}& \lesssim \left(\sum_{N\leq 1}+\sum_{N\geq
1}\right)\Vert P_{N}^{\eta }\mathfrak{m}(\xi ,\eta )\Vert _{L_{\xi }^{\infty
} \dot{H}_{\eta }^{s}} \\
& \lesssim \sum_{N\leq 1}N^{\varepsilon }\Vert \mathfrak{m}\Vert _{L_{\xi
}^{\infty }\dot{H}_{\eta }^{s-\varepsilon }}+\sum_{N\geq 1}N^{-\varepsilon
}\Vert \mathfrak{m}\Vert _{L_{\xi }^{\infty }\dot{H}_{\eta }^{s+\varepsilon
}}<+\infty .
\end{split}%
\end{equation*}
\end{proof}

\begin{theorem}
\label{CorGNT} Suppose that $0\le s\le n/2$ and $\Vert \mathfrak{m}\Vert
_{M_{\eta ,\xi }^{s,\infty }}=\Vert \mathfrak{m}\Vert _{M_{\eta ,\xi
}^{s}}<\infty ,$ then 
\begin{equation}
\Vert B[f,g]\Vert _{L^{l_{1}^{\prime }}}\lesssim \Vert \mathfrak{m}\Vert
_{M_{\eta ,\xi }^{s}}\Vert f\Vert _{L^{l_{2}}}\Vert g\Vert _{L^{2}},
\end{equation}%
for $l_{1},l_{2}$ satisfying 
\begin{equation}
2\leq l_{1},l_{2}\leq \frac{2n}{n-2s}\text{ and }\frac{1}{l_{1}}+ \frac{1}{%
l_{2}}=1-\frac{s}{n}.  \label{CondpqBilEstGNT}
\end{equation}
\end{theorem}

\begin{remark}
Actually, by changing coordinates $(\xi ,\eta )$ to $(\xi ,\zeta =\xi -\eta
) $, we could replace the norm $M_{\xi ,\eta }^{s}$ by 
\begin{equation*}
\min \left( \Vert \mathfrak{m}\Vert _{M_{\xi ,\eta }^{s}},\Vert \mathfrak{m}%
\Vert _{M_{\xi ,\zeta }^{s}}\right) .
\end{equation*}
\end{remark}

Theorem \ref{CorGNT} follows by duality from the following estimate which is
an adaptation of an estimate from Gustafson, Nakanishi and Tsai \cite{GNT}.

\begin{lemma}
\label{LemGNT} Let $0\le s\le n/2$, $2\leq l_{1},l_{2},l_{3}\leq \frac{2n}{%
n-2s}$, then 
\begin{equation}
\Vert B[f,g]\Vert _{L^{l_{1}}}\lesssim \Vert \mathfrak{m}\Vert _{\mathcal{M}
_{\xi ,\eta }^{s,b}}\Vert f\Vert _{L^{l_{2}}}\Vert g\Vert _{L^{l_{3}}}
\label{BilEstGNT}
\end{equation}
for all $f\in L^{l_2}$, $g\in L^{l_3}$, where $\frac{1}{b}+\frac{1}{l_{1}}=%
\frac{1}{2}$, $\frac{1}{l_{2}}+\frac{1}{ l_{3}}=1-\frac{s}{n}.$
\end{lemma}

\begin{proof}
We consider $\mathfrak{m}$ with finite $\mathcal{M}_{\xi ,\eta }^{s,b}$
norm. Let $\mathcal{F}_{x}^{\eta }$ denote the Fourier transform from $%
x\rightarrow \eta .$ By definition, we have 
\begin{equation}
\hat{f}(\eta )\hat{g}(\xi -\eta )=\mathcal{F}_{x}^{\eta }\mathcal{F}%
_{y}^{\xi }f(x+y)g(y)  \label{DoubleFT}
\end{equation}%
and we let $\mathfrak{m}_{N}(\xi ,\eta )=P_{N}^{\eta }\mathfrak{m}(\xi ,\eta
)$ so that $\mathcal{F}_{z}^{\eta }\mathfrak{m}_{N}(\xi ,\eta )=\chi(\frac{z 
}{N})\mathcal{F}_{\eta }^{z}\mathfrak{m}_N(\xi ,\eta )$. Using first
Parseval's equality in $x$, then in $\eta $ and then in $\xi$, we see that 
\begin{equation*}
\begin{split}
\int_{\mathbb{R}^{n}}B[f,g](x)h(x)dx& =\int_{\mathbb{R}^{2n}}\mathfrak{m}
_{N}(\xi ,\eta )\hat{h}(\xi )\hat{f}(\eta )\hat{g}(\xi -\eta )d\eta d\xi \\
& =\int_{\mathbb{R}^{n}}\hat{h}(\xi )\int_{\mathbb{R}^{n}}\mathfrak{m}
_{N}(\xi ,\eta )\left( \mathcal{F}_{x}^{\eta }\mathcal{F}_{y}^{\xi
}f(x+y)g(y)\right) d\eta d\xi \\
& =\int_{\mathbb{R}^{n}}\hat{h}(\xi )\int_{\mathbb{R}^{n}}\mathcal{F}
_{x}^{\eta }\mathfrak{m}_{N}(\xi ,\eta )\left( \mathcal{F}_{y}^{\xi
}f(x+y)g(y)\right) d\eta d\xi \\
& =\int_{\mathbb{R}^{n}}\hat{h}(\xi )\int_{\mathbb{R}^{n}}\left( \chi(\frac{x%
}{N})\mathcal{F}_{\eta }^{x}\mathfrak{m}_N(\xi ,\eta )\right) \left(\mathcal{%
F}_{y}^{\xi }f(x+y)g(y)\right) dxd\xi \\
& =\int_{\mathbb{R}^{n}}\hat{h}(\xi )\int_{\mathbb{R}^{n}}\left(\mathcal{F}%
_{\eta }^{x}\mathfrak{m}_N(\xi ,\eta )\right) \mathcal{F} _{y}^{\xi} \left(
\chi (\frac{ x}{N})f(x+y)g(y)\right) dxd\xi
\end{split}%
\end{equation*}

We then use Cauchy-Schwarz's inequality for the inner integral for $x,$ and
then use the H\"{o}lder inequality with $\frac{1}{a}+\frac{1}{b}=\frac{1}{2}$
to get 
\begin{eqnarray*}
&&\int_{\mathbb{R}^{n}}|\hat{h}(\xi )\vert \Vert \mathcal{F} _{\eta }^{x}%
\mathfrak{m}_N(\xi ,\eta )\Vert_{L_{x}^{2}}(\xi )\Vert \mathcal{F} _{y}^{\xi
}\{\chi (\frac{ x}{N}) \left( f(x+y)g(y)\right) \}\Vert_{L_{x}^{2}}(\xi )d\xi
\\
&\leq &\Vert\hat{h}\Vert_{L^{a}_\xi}\Vert \mathfrak{m}_N(\xi ,\eta
)\Vert_{L_{\xi }^{b}(L_{\eta }^{2})}\Vert\mathcal{F}_{y}^{\xi }\{\chi (\frac{
x}{N})\left( f(x+y)g(y)\right) \}\Vert_{L_{x,\xi }^{2}} \\
&\leq &\Vert h\Vert_{L^{a^{\prime }}_x}\Vert\mathfrak{m}_N(\xi ,\eta
)\Vert_{L_{\xi }^{b}(L_{\eta }^{2})}\Vert \chi (\frac{ x}{N}%
)f(x+y)g(y)\Vert_{L_{x,y}^{2}},
\end{eqnarray*}
where we have used the Hausdroff-Young's inequality for $a>2,$ and the
Parseval's equality in $\eta $ for the second factor, as well as the
Parseval's equality in $\xi $ for the third factor. Finally, since $%
\Vert\chi (\frac{ x}{N})\Vert_{L^{n/s}}\lesssim N^{s},$ we employ the
Hardy-Littlewood Young's inequality with $\frac{s}{n}+\frac{1}{l_{2}}+\frac{1%
}{l_{3}}=1$ to get that 
\begin{equation*}
\Vert \chi (\frac{ x}{N})f(x+y)g(y)\Vert _{L_{x,y}^{2}}\lesssim N^{s}\Vert
f\Vert _{L^{l_{1}}}\Vert g\Vert _{L^{l_{2}}}.
\end{equation*}
Combining $N^{s}$ with $||\mathfrak{m}(\xi ,\eta )||_{L_{\xi }^{b}(L_{\eta
}^{2})}$ with \eqref{CondpqBilEstGNT}, we complete the proof.
\end{proof}

In order to prove theorem \ref{CorGNT}, it suffices to remark that 
\begin{equation*}
\begin{split}
\int_{\mathbb{R}^{n}}B[f,g](x)h(x)dx& =\int_{\mathbb{R}^{n}}\mathcal{F}%
_{x}^{\xi }B[f,g](\xi )\hat{h}(\xi )d\xi \\
& =\int_{\mathbb{R}^{2n}}\hat{f}(\eta )\mathfrak{m}(\xi ,\eta )\hat{h}(\xi )%
\hat{g}(\xi -\eta )d\xi d\eta \\
& =\int_{\mathbb{R}^{n}}f(x)B^{\ast }[h,\bar{g}](x)dx.
\end{split}%
\end{equation*}%
Applying \eqref{BilEstGNT} to $B^{\ast }$ with $l_{1}=2$ to the bilinear
operator corresponding to the multiplier $\mathfrak{m}^{\ast }(\xi ,\eta )=%
\mathfrak{m}(\eta ,\xi )$, we get the Theorem.

\subsection{Multiplier Analysis}

The control of $L^{10}$ norm is the main mathematical difficulty in this
paper. In this subsection, we prove the relevant estimate to apply Theorem %
\ref{CorGNT} to the multipliers that appear in our analysis.

\begin{lemma}
\label{EstimPhiGen} Let $a=b+c\in \mathbb{R}^{3},$ and let $|c|\leq \min
\{|a|,|b|\},$ then%
\begin{equation}
|p(a)-p(b)-p(c)|\gtrsim |c|\{1-\cos [c,a]+1-\cos [b,a]\}+\frac{|a||b||c|}{%
(1+|a||b|)(1+|c|^{2})}.  \label{EstimOnPhi}
\end{equation}%
where $[\cdot ,\cdot ]$ denote the angle between two vectors.
\end{lemma}

\begin{proof}
We first note that if $|b|\geq |a|,$ then $p(b)\geq p(a)$ and 
\begin{equation*}
|p(a)-p(b)-p(c)|\geq p(c)\gtrsim |c|
\end{equation*}%
and the lemma follows. We assume $|b|\leq |a|.$ We remark that, as written
in \eqref{DefOfp}, $p(r)=rq(r)$, where $1\leq q(r)\leq q(0)=\sqrt{2}$ and 
\begin{equation}
\begin{split}
q^{\prime }(r)& =-\frac{r}{(1+r^{2})^{2}\sqrt{\frac{2+r^{2}}{1+r^{2}}}}\sim
_{r\rightarrow \infty }-\frac{1}{r^{3}} \\
q^{\prime }(0)& =0,\hskip.3cmq^{\prime \prime }(0)=-\frac{1}{\sqrt{2}}.
\end{split}
\label{EstimOnQ}
\end{equation}%
From this, we get that 
\begin{equation}
p(a)-p(b)-p(c)\leq \left[ |a|-|b|-|c|\right] q(a)-|b|\left( q(b)-q(a)\right)
-|c|\left( q(c)-q(a)\right) .  \label{EstonPhiABCGenPos}
\end{equation}%
From \eqref{EstimOnQ}, we see that $q$ is decreasing and hence each term is
non positive. Remarking that $|a|=|b|\cos [b,a]+|c|\cos [c,a]$, the first
term above gives the first term on the right hand side in \eqref{EstimOnPhi}.

We now consider the last term in the right hand side. Notice first that if $%
\cos [c,a]\leq 9/10$, then the last term is bounded by $\vert c\vert(1-\cos[%
c,a])$ and the lemma is clearly valid. So we can assume that $c$ and $a$ are
almost collinear with $\cos [c,a]\geq 9/10 $. In which case, we get that $%
|a|\geq 4/3|c|$ and 
\begin{equation*}
|a|-|c|\sim |b|\sim |a|.
\end{equation*}%
Using \eqref{EstimOnQ}, we see that there exists $\delta >0$ such that $%
-s\leq q^{\prime }(s)\leq -\frac{s}{2}$ for $0\leq s\leq \delta $.
Consequently, if $|a|\leq \delta $, we get that 
\begin{equation*}
\begin{split}
|c|(q(a)-q(c))& =|c|\int_{|c|}^{|a|}q^{\prime }(s)ds\leq -|c|\frac{%
|a|^{2}-|c|^{2}}{4}\lesssim -|a||c|(|a|-|c|)\lesssim-|a||b||c|.
\end{split}%
\end{equation*}
On the other hand, since $q^{\prime }(r)\sim -r^{-3}$ at $\infty $, we see
that 
\begin{equation*}
q(a)-q(c)=\int_{|c|}^{|a|}q^{\prime }(s)ds\sim _{\infty }-\int_{|c|}^{|a|} 
\frac{ds}{s^{3}}=\frac{|c|^{2}-|a|^{2}}{2|a|^{2}|c|^{2}}\lesssim -\frac{
|a|-|c|}{|a||c|^{2}}
\end{equation*}
so that if $|c|\geq \delta ^{-1}$ is sufficiently large, we get that 
\begin{equation*}
|c|(q(a)-q(c))\lesssim-\frac{1}{|c|}.
\end{equation*}
Finally, in the last case $\delta\le\vert a\vert\le\delta^{-1}$ and $%
|a|=|c|+(|a|-|c|)\geq |c|+\delta /2$. Therefore, 
\begin{equation*}
\int_{|c|}^{|a|}q^{\prime }(s)ds\lesssim \int_{|c|}^{|c|+\delta /2}q^{\prime
}(s)ds\lesssim -\frac{\delta }{2}q^{\prime }(2\delta ^{-1})
\end{equation*}
and we recover the last term once again.
\end{proof}

In the remaining part of this section, we consider the triangle with
vertices $\xi ,\eta ,\xi -\eta $ and let $\theta $ be the angle between $\xi 
$ and $\eta $ ($0\leq \theta \leq \pi $), $\gamma $ the angle between $\xi $
and $\xi -\eta $ ($0\leq \gamma \leq \pi $) and we let the angle between $%
\eta $ and $\eta -\xi $ by $\pi -\beta $ such that $\beta =\gamma +\theta .$
We note that $\sin \frac{\beta }{2}\leq \frac{\beta }{2}$ and $\sin \frac{%
\beta }{2}\backsim \beta $ for $0\leq \beta \leq \pi $ so that 
\begin{equation*}
1-\cos \beta =2\sin ^{2}\frac{\beta }{2}\backsim \beta ^{2}.
\end{equation*}

We now obtain general bounds on the multipliers that arise in our analysis.
We first focus on the multiplier associated with the phase $\Phi_1$. In the
end, in Section \ref{SecL10}, we recover the bounds on the other multipliers
using symmetry.

\begin{lemma}
\label{EstimDerOfPhi} The following estimates on $\Phi_1$ are globally true: 
\begin{align}
|\partial _{\xi }\Phi _{1}(\xi ,\eta )|&\lesssim \frac{|\eta |}{\langle \max
\{|\xi -\eta |,|\xi |\}\rangle \langle \min \{|\xi -\eta |,|\xi |\}\rangle
^{2}}+|\sin \gamma |,\text{ \ \ \ \ \ }  \label{phixi} \\
|\partial _{\eta }\Phi _{1}(\xi ,\eta )|& \lesssim \frac{|\xi |}{\langle
\max \{|\xi -\eta |,|\eta |\}\rangle \langle \min \{|\xi -\eta |,|\eta
|\}\rangle ^{2}}+|\sin \beta |,  \label{phieta} \\
|\Delta _{\xi }\Phi _{1}(\xi ,\eta )|& \lesssim \frac{|\eta |}{\langle \max
\{|\xi -\eta |,|\xi |\}\rangle \langle \min \{|\xi -\eta |,|\xi |\}\rangle
^{3}}+\frac{|\eta |}{|\xi -\eta ||\xi |} ,\text{ \ \ \ \ \ }  \label{phixixi}
\\
|\Delta _{\eta }\Phi _{1}(\xi ,\eta )|& \lesssim \frac{1}{%
\min(\vert\xi-\eta\vert,\vert\eta\vert)}  \label{phietaeta}
\end{align}
for all $\xi,\eta\in\mathbb{R}^3$.
\end{lemma}

\begin{proof}
Recall $\Phi_1=p(\xi)-p(\xi-\eta)-p(\eta)$. We compute 
\begin{equation*}
\begin{split}
\left\vert \nabla_{\xi }\Phi _{1}\right\vert & =\left\vert p^{\prime }(\xi )%
\frac{\xi }{|\xi |}-p^{\prime }(\xi -\eta )\frac{\xi -\eta }{|\xi -\eta |}%
\right\vert \\
& \leq \left\vert p^{\prime }(\xi )-p^{\prime }(\xi -\eta )\right\vert
+|p^{\prime }(\xi )|\left\vert \frac{\xi }{|\xi |}-\frac{\xi -\eta }{|\xi
-\eta |}\right\vert \\
& \lesssim
\left\vert\int_{\vert\xi-\eta\vert}^{\vert\xi\vert}p^{\prime\prime}(s)ds%
\right\vert +\left\vert \frac{\xi }{|\xi |}-\frac{\xi -\eta }{|\xi -\eta |}%
\right\vert \\
& \lesssim
\left\vert\int_{\vert\xi-\eta\vert}^{\vert\xi\vert}p^{\prime\prime}(s)ds%
\right\vert+2\sin \frac{\gamma }{2}.
\end{split}%
\end{equation*}
We claim that 
\begin{equation}  \label{DifferenceinPPrime}
\left\vert p^\prime(\xi)-p^\prime(\xi-\eta)\right\vert\lesssim\frac{|\eta |}{%
\langle \max \{|\xi -\eta |,|\xi |\}\rangle \langle \min \{|\xi -\eta |,|\xi
|\}\rangle ^{2}}
\end{equation}

In fact, if $\max(\vert\xi\vert,\vert\xi-\eta\vert)\le 20$, from %
\eqref{Estimp}, using the crude bound $\vert
p^{\prime\prime}(s)\vert\lesssim 1$, we obtain that 
\begin{equation*}
\left\vert\int_{\vert\xi-\eta\vert}^{\vert\xi\vert}p^{\prime\prime}(s)ds%
\right\vert\lesssim \left\vert \vert\xi\vert-
\vert\xi-\eta\vert\right\vert\lesssim \vert\eta\vert\lesssim \frac{|\eta |}{%
\langle \max \{|\xi -\eta |,|\xi |\}\rangle \langle \min \{|\xi -\eta |,|\xi
|\}\rangle ^{2}}.
\end{equation*}
Therefore, we only need to consider the case $\max\{\vert\xi\vert,\vert\xi-%
\eta\vert\}\ge 20$. Then, if $\min\{\vert\xi\vert,\vert\xi-\eta\vert\}\le 10$%
, we get that $\vert\eta\vert\simeq\max\{\vert\xi\vert,\vert\xi-\eta\vert\}$
and the right-hand side of \eqref{phixi} is of order $1$ and the claim is
valid. Finally, if $\min\{\vert\xi\vert,\vert\xi-\eta\vert\}\ge 10$, from %
\eqref{Estimp}, $p^{\prime \prime }(r)\sim \frac{1}{r^{3}}$ as $r\rightarrow
\infty ,$ and we conclude that claim since 
\begin{equation}  \label{min}
\begin{split}
\left\vert\int_{\vert\xi-\eta\vert}^{\vert\xi\vert}p^{\prime\prime}(s)ds%
\right\vert
&\lesssim\left\vert\int_{\min\{\vert\xi\vert,\vert\xi-\eta\vert\}}^{\max\{%
\vert \xi\vert,\vert\xi-\eta\vert\}} \frac{1}{ r^{3}}dr\right\vert \lesssim 
\frac{1}{\min\{\vert\xi\vert,\vert\xi-\eta\vert\}^{2}}-\frac{1}{%
\max\{\vert\xi\vert ,\vert \xi-\eta\vert\}^2} \\
&=\frac{\left\vert \vert\xi\vert- |\xi -\eta
|\right\vert(\vert\xi\vert+\vert \xi -\eta\vert)}{\vert\xi\vert^{2}|\xi
-\eta |^{2}}\lesssim \frac{\vert\eta\vert}{\min\{\vert\xi\vert,\vert\xi-\eta%
\vert\}^{2}\max\{\vert \xi\vert,\vert \xi -\eta\vert\}} \\
&\lesssim \frac{|\eta |}{\langle \max \{|\xi -\eta |,|\xi |\}\rangle \langle
\min \{|\xi -\eta |,|\xi |\}\rangle ^{2}}.
\end{split}%
\end{equation}

Similarly, as in \eqref{DifferenceinPPrime}, 
\begin{equation*}
\begin{split}
\left\vert \nabla _{\eta }\Phi _{1}\right\vert & =\left\vert -p^{\prime
}(\xi -\eta )\frac{\xi -\eta }{|\xi -\eta |}-p^{\prime }(\eta )\frac{\eta }{%
|\eta |}\right\vert \\
& \leq \left\vert p^{\prime }(\eta )-p^{\prime }(\xi -\eta )\right\vert
+|p^{\prime }(\xi )|\left\vert \frac{\eta }{|\eta |}-\frac{\xi -\eta }{|\xi
-\eta |}\right\vert \\
& \leq \left\vert p^{\prime }(\eta )-p^{\prime }(\xi -\eta )\right\vert +2%
\sqrt{2}\sin \frac{\beta }{2} \\
&\lesssim\frac{|\xi |}{\langle \max \{|\xi -\eta |,|\eta |\}\rangle \langle
\min \{|\xi -\eta |,|\eta |\}\rangle ^{2}}+|\sin \beta |.
\end{split}%
\end{equation*}
Using the fact that $p^{(3)}(r)\sim r^{-4}$ as $r\to+\infty$, we now
compute, by \eqref{DifferenceinPPrime}, that 
\begin{equation*}
\begin{split}
\vert\Delta _{\xi}\Phi _{1}\vert & =\vert \Delta p(\xi )-\Delta p(\xi -\eta
)\vert \\
& =\left\vert p^{\prime\prime}(\xi)-p^{\prime \prime}(\xi -\eta ) +2\left( 
\frac{p^{\prime}(\xi)}{\vert\xi\vert}-\frac{p^{\prime}(\xi -\eta)}{\vert \xi
-\eta \vert}\right) \right\vert \\
& \lesssim \frac{|\eta |}{\langle \max \{|\xi -\eta |,|\xi |\}\rangle
\langle \min \{|\xi -\eta |,|\xi |\}\rangle ^{3}} +\frac{\vert p^{\prime
}(\xi )-p^{\prime }(\xi -\eta )\vert}{\vert\xi\vert} \\
&+\left\vert\frac{1}{|\xi |}-\frac{1}{|\xi -\eta |}\right\vert \vert
p^{\prime }(\xi -\eta )\vert \\
& \lesssim \frac{|\eta |}{\langle \max \{|\xi -\eta |,|\xi |\}\rangle
\langle \min \{|\xi -\eta |,|\xi |\}\rangle ^{3}} \\
&+\frac{|\eta |}{\langle \max \{|\xi -\eta |,|\xi |\}\rangle \langle \min
\{|\xi -\eta |,|\xi |\}\rangle ^{2}|\xi |}+\frac{|\eta |}{\vert \xi -\eta
\vert\vert \xi \vert }.
\end{split}%
\end{equation*}
Finally, we also get that 
\begin{equation*}
\begin{split}
\Delta _{\eta }\Phi _{1}& =-\Delta _{\eta }p(\xi -\eta )-\Delta _{\eta
}p(\eta ) \\
& =-\left( p^{\prime \prime }(\xi -\eta )+p^{\prime \prime }(\eta )\right)
-\left( \frac{2}{|\xi -\eta |}p^{\prime }(\xi -\eta )+\frac{2}{|\eta |}%
p^{\prime }(\eta )\right) \\
& \lesssim \frac{1}{1+|\xi -\eta |^{3}}+\frac{1}{1+|\eta |^{3}}+\frac{1}{%
|\eta |}+\frac{1}{|\xi -\eta |}.
\end{split}%
\end{equation*}
This ends the proof.
\end{proof}

\begin{proposition}
\label{EstimGenPhase} Define 
\begin{equation*}
\mathfrak{M}_{1}=\frac{|\xi ||\xi -\eta ||\eta |}{\Phi _{1}\langle \xi -\eta
\rangle ^{2\lambda }\langle \eta \rangle ^{2\lambda }}
\end{equation*}
then if $f$ is either $\frac{\chi }{\langle \xi -\eta \rangle ^{\frac{1}{2}}}
$ or $\frac{\chi }{\langle \eta \rangle ^{\frac{1}{2}}}$ for any cutoff
function $\chi $ with support in 
\begin{equation}  \label{RegionOmega}
\Omega =\{\max \{|\xi |,|\xi -\eta |,|\eta |\}\gtrsim 1\},
\end{equation}
we have that, for any $\varepsilon >0,$ 
\begin{eqnarray}
||\mathfrak{M}_{1}||_{L_{\eta }^{\infty }(H_{\xi }^{\frac{5}{4}-\varepsilon
})}+||\mathfrak{M}_{1}||_{L_{\xi }^{\infty }(H_{\eta }^{\frac{5}{4}%
-\varepsilon })} &\lesssim_{\varepsilon}1\text{\ \ }for\text{ \ \ }\lambda >%
\frac{9}{8}.  \label{mglobal} \\
||f\mathfrak{M}_{1}||_{L_{\eta }^{\infty }(H_{\xi }^{\frac{3}{2}-\varepsilon
})}+||f\mathfrak{M}_{1}||_{L_{\xi }^{\infty }(H_{\eta }^{\frac{3}{2}%
-\varepsilon })} &\lesssim_{\varepsilon}1\text{ \ \ }for\text{ }\lambda >1.
\label{mlocal}
\end{eqnarray}
\end{proposition}

\begin{proof}[Proof of Proposition \protect\ref{EstimGenPhase}]
In order to prove this proposition, we split $\mathbb{R}^3$ into a union of
three regions: $\{|\xi |<\frac{1}{2} |\eta |\},$ $\{|\eta |<\frac{1}{2}|\xi
|\}$ and $\{\frac{1}{3}<\frac{|\xi |}{ |\eta |}<3\}.$ Before we start, we
remark that, in the triangle defined by $\xi ,\eta $ and $\xi -\eta ,$ we
have that 
\begin{equation*}
\frac{|\eta -\xi |}{\sin \theta }=\frac{|\xi |}{\sin \beta }=\frac{|\eta |}{
\sin \gamma }.
\end{equation*}
\textbf{Case 1. The region} $\Omega _{1}=\{|\xi |<\frac{1}{2}|\eta |\}.$ In
this case, $|\xi -\eta |\geq |\eta |-|\xi |>|\xi |,$ so $|\xi |$ has the
smallest size. We also deduce that $|\xi -\eta |\simeq |\eta |$ and
consequently, since $p(\xi)\le p(\eta)$ 
\begin{equation*}
|\Phi _{1}(\xi ,\eta )|=|p(\xi )-p(\xi -\eta )-p(\eta )|\gtrsim \max \{|\eta
|,|\xi -\eta |\}.
\end{equation*}
We note that since $p^{\prime }$ is bounded, $|\nabla _{\xi ,\eta }\Phi
_{1}|\lesssim 1$ and from Lemma \ref{EstimDerOfPhi}, we obtain that 
\begin{eqnarray*}
|\nabla _{\xi ,\eta }\{\frac{1}{\Phi _{1}}\}| &=&\left\vert -\frac{\nabla
_{\xi ,\eta }\Phi _{1}}{\Phi _{1}^{2}}\right\vert \lesssim \frac{1}{\{|\eta
|+|\xi -\eta |\}^{2}}, \\
|\Delta _{\xi }\{\frac{1}{\Phi _{1}}\}| &=&\left\vert -\frac{\Delta _{\xi
}\Phi _{1}}{\Phi _{1}^{2}}+2\frac{|\nabla _{\xi }\Phi _{1}|^{2}}{\Phi
_{1}^{3}}\right\vert \\
&\lesssim &\frac{1}{\{|\eta |+|\xi -\eta |\}^{2}}\left\{ \frac{1}{ \{1+|\xi
|^{3}\}}+\frac{1}{|\xi|}\right\} +\frac{1}{\{|\eta |+|\xi -\eta |\}^{3}} \\
&\lesssim & \frac{1}{|\xi |^{3}}, \\
|\Delta _{\eta }\{\frac{1}{\Phi _{1}}\}| &=&\left\vert -\frac{\Delta _{\eta
}\Phi _{1}}{\Phi _{1}^{2}}+2\frac{|\nabla _{\eta }\Phi _{1}|^{2}}{\Phi
_{1}^{3}}\right\vert \\
&\lesssim &\frac{1}{\{|\eta |+|\xi -\eta |\}^{2}}\frac{1}{|\eta |} +\frac{1}{%
\{|\eta |+|\xi -\eta |\}^{3}} \lesssim \frac{1}{|\eta |^{3}}.
\end{eqnarray*}
Recall the definition of $\chi,\varphi$ from \eqref{DefLitPalOp} and denote $%
g=\frac{|\xi ||\xi -\eta ||\eta |}{\langle \xi -\eta \rangle ^{2\lambda
}\langle \eta \rangle ^{2\lambda }}\varphi (\frac{\xi }{N} )\chi (\frac{%
2|\xi |}{|\eta |}),$ so that 
\begin{eqnarray*}
\left\vert\mathfrak{M}_{1}\varphi (\frac{\xi}{N})\chi (\frac{2|\xi |}{|\eta |%
} )\right\vert&\lesssim & \frac{1}{\langle N\rangle ^{4\lambda-2}}\varphi (%
\frac{\xi }{N}),\hskip.1cm\hbox{and} \\
\left\vert \Delta _{\xi }\{\mathfrak{M}_{1}\varphi (\frac{\xi}{N})\chi (%
\frac{2|\xi |}{ |\eta |})\}\right\vert &=&\left\vert \Delta _{\xi }\{\frac{1%
}{\Phi _{1}}\}g+2\nabla _{\xi }\{ \frac{1}{\Phi _{1}}\}\cdot \nabla _{\xi }g+%
\frac{1}{\Phi _{1}}\Delta _{\xi }g\right\vert \\
&\lesssim &\frac{1}{N^{2}\langle N\rangle ^{4\lambda-2}}\varphi(\frac{\xi}{N}%
).
\end{eqnarray*}
We thus have that 
\begin{eqnarray*}
||\mathfrak{M}_{1}\varphi (\frac{\xi}{N})\chi (\frac{2|\xi |}{|\eta |}
)||_{L_{\xi }^{2}} &\lesssim &\frac{N^{3/2}}{\langle N\rangle ^{4\lambda-2}}
\\
\Vert\Delta _{\xi }\{\mathfrak{M}_{1}\varphi (\frac{\xi}{N})\chi (\frac{%
2|\xi | }{|\eta |})\}\Vert_{L^{2}}+\Vert\Delta _{\eta }\{\mathfrak{M}%
_{1}\varphi (\frac{\eta }{N})\chi (\frac{2|\xi |}{|\eta |})\}\Vert_{L^{2}}
&\lesssim &\frac{1}{ N^{1/2}\langle N\rangle ^{4\lambda-2}}.
\end{eqnarray*}
Interpolating between the above estimates, we get that for any $\varepsilon
>0$ and any fixed fixed $\eta$, 
\begin{eqnarray*}
\Vert \mathfrak{M}_{1}\chi (\frac{2|\xi |}{ |\eta |})\Vert_{\dot{H}_{\xi
}^{\sigma}}&\lesssim&\sum_{N}\Vert\mathfrak{M}_{1}\varphi (\frac{\xi}{N}%
)\chi (\frac{2|\xi |}{|\eta |})\Vert_{\dot{H}_{\xi }^{\sigma }} \\
&\lesssim &\sum_{N}||\mathfrak{M}_{1}\varphi (\frac{\xi}{N})\chi (\frac{
2|\xi |}{|\eta |})\Vert_{L^{2}}^{1-\frac{\sigma }{2}} \Vert\Delta _{\xi }\{%
\mathfrak{M}_{1}\varphi (\frac{\xi}{N})\chi (\frac{2|\xi |}{|\eta |}
)\}\Vert_{L^{2}}^{\frac{\sigma }{2}} \\
&\lesssim &\sum_{N}\frac{N^{\frac{3}{2}-\sigma }}{\langle N\rangle
^{4\lambda-2}},
\end{eqnarray*}
which is summable in $N$ for $\lambda >1$ and $0\le\sigma <\frac{3}{2}.$ The
same proof (switching $\xi $ to $\eta $) works for $\sum_{N}||\mathfrak{M}%
_{1}\varphi (\frac{\eta}{N})\chi (\frac{2|\xi |}{|\eta |})||_{\dot{H}_{\eta
}^{\sigma }}.$ Both (\ref{mglobal}) and (\ref{mlocal}) are valid in this
case.

\medskip

\textbf{Case 2. In the region }$\Omega _{2}=\{|\eta |\leq \frac{1}{2}|\xi
|\}.$ \textbf{\ }We note that $|\eta |$ is the smallest, and $|\xi -\eta
|\simeq |\xi |.$ We first claim that 
\begin{equation}
|\Phi _{1}|\gtrsim |\eta |\{\theta ^{2}+\frac{|\xi |^{2}}{\langle \eta
\rangle ^{2}\langle \xi \rangle ^{2}}\}\equiv |\eta |(\theta ^{2}+d^{2}).
\label{d}
\end{equation}%
In fact, if $|\xi |$ is not the largest, then we know that $|\Phi _{1}|\geq
|\eta |$ and the claim is clearly valid. If $|\xi |$ is the largest, then $%
\theta $ is the angle between $|\xi |$ and $|\eta |,$ and $1-\cos \theta
\gtrsim \theta ^{2}$. Therefore we deduce \eqref{d} from Lemma \ref%
{EstimPhiGen}.

We note from Lemma \ref{EstimDerOfPhi} that in this case, 
\begin{align*}
& |\nabla _{\xi }\Phi _{1}(\xi ,\eta )|\lesssim \frac{|\eta |}{1+|\xi |^{3}}%
+|\sin \gamma |, \\
& |\Delta _{\xi }\Phi _{1}(\xi ,\eta )|\lesssim \frac{|\eta |}{1+|\xi |^{4}}+%
\frac{|\eta |}{|\xi |^{2}}\lesssim \text{ }\frac{|\eta |}{|\xi |^{2}}.
\end{align*}%
Besides, using that $\frac{\sin \gamma }{|\eta |}=\frac{\sin \beta }{|\xi |}$%
, the inequality above and \eqref{d}, we can obtain that 
\begin{equation}
\begin{split}
|\nabla _{\xi }\{\frac{1}{\Phi _{1}}\}|& =\left\vert -\frac{\nabla _{\xi
}\Phi _{1}}{\Phi _{1}^{2}}\right\vert \lesssim \frac{1}{|\eta |\{\theta
^{2}+d^{2}\}^{2}(1+|\xi |^{3})}+\frac{\sin \beta }{|\eta ||\xi |\{\theta
^{2}+d^{2}\}^{2}}, \\
|\Delta _{\xi }\{\frac{1}{\Phi _{1}}\}|& =\left\vert -\frac{\Delta _{\xi
}\Phi _{1}}{\Phi _{1}^{2}}+2\frac{|\nabla _{\xi }\Phi _{1}|^{2}}{\Phi
_{1}^{3}}\right\vert \\
& \lesssim \frac{1}{|\eta |\{\theta ^{2}+d^{2}\}^{2}}\frac{1}{|\xi |^{2}}+%
\frac{\frac{|\eta |^{2}}{(1+|\xi |^{3})^{2}}+\sin ^{2}\gamma }{|\eta
|^{3}\{\theta ^{2}+d^{2}\}^{3}} \\
& \lesssim \frac{1}{|\eta |\{\theta ^{2}+d^{2}\}^{2}|\xi |^{2}}+\frac{1}{%
(1+|\xi |^{3})^{2}|\eta |\{\theta ^{2}+d^{2}\}^{3}}+\frac{\sin ^{2}\beta }{%
|\eta ||\xi |^{2}\{\theta ^{2}+d^{2}\}^{3}}.
\end{split}
\label{ComputationOfNablaxi}
\end{equation}%
Now, for fixed $\eta ,$ and for any cutoff function for $\chi (\frac{2|\eta |%
}{|\xi |}),$ denote 
\begin{equation*}
g=\frac{|\xi ||\xi -\eta ||\eta |}{\langle \xi -\eta \rangle ^{2\lambda
}\langle \eta \rangle ^{2\lambda }}\varphi (\frac{\xi }{N})\chi (\frac{%
2|\eta |}{|\xi |})
\end{equation*}%
since $|\xi |\backsim N$ and $|\eta |\leq |\xi |,$ direct computation yields 
\begin{eqnarray*}
|\nabla _{\xi }g| &\lesssim &\frac{1}{|\xi |}\frac{|\xi ||\xi -\eta ||\eta |%
}{\langle \xi -\eta \rangle ^{2\lambda }\langle \eta \rangle ^{2\lambda }}%
\mathbf{1}_{|\xi |\backsim N,|\eta |\leq |\xi |}, \\
|\partial _{\xi }^{2}g| &\lesssim &\frac{1}{|\xi |^{2}}\frac{|\xi ||\xi
-\eta ||\eta |}{\langle \xi -\eta \rangle ^{2\lambda }\langle \eta \rangle
^{2\lambda }}\mathbf{1}_{|\xi |\backsim N,|\eta |\leq |\xi |}.
\end{eqnarray*}%
Therefore, we have that 
\begin{equation*}
|\mathfrak{M}_{1}\varphi (\frac{\xi }{N})\chi (\frac{2|\xi |}{|\eta |}%
)|\lesssim \frac{|\xi |^{2}}{\langle N\rangle ^{2\lambda }\langle \eta
\rangle ^{2\lambda }(\theta ^{2}+d^{2})}\mathbf{1}_{|\xi |\backsim N,|\eta
|\leq |\xi |},
\end{equation*}%
and by \eqref{ComputationOfNablaxi} we also have that 
\begin{eqnarray*}
&&|\Delta _{\xi }\{\mathfrak{M}_{1}\varphi (\frac{\xi }{N})\chi (\frac{2|\xi
|}{|\eta |})\}| \\
&=&|\Delta _{\xi }\{\frac{1}{\Phi _{1}}\}g+2\nabla _{\xi }\{\frac{1}{\Phi
_{1}}\}\cdot \nabla _{\xi }g+\frac{1}{\Phi _{1}}\Delta _{\xi }g| \\
&\lesssim &\frac{|\xi |^{2}\mathbf{1}_{|\xi |\backsim N,|\eta |\leq |\xi |}}{%
\langle N\rangle ^{2\lambda }\langle \eta \rangle ^{2\lambda }}\left\{ \frac{%
1}{\{\theta ^{2}+d^{2}\}^{2}|\xi |^{2}}+\frac{1}{(1+|\xi |^{3})^{2}\{\theta
^{2}+d^{2}\}^{3}}+\frac{\sin ^{2}\beta }{|\xi |^{2}\{\theta ^{2}+d^{2}\}^{3}}%
\right\} .
\end{eqnarray*}%
By using $\frac{\eta }{|\eta |}$ as the north pole, we thus compute: 
\begin{equation}
\begin{split}
\int |\mathfrak{M}_{1}\varphi (\frac{\xi }{N})\chi (\frac{2|\xi |}{|\eta |}%
)|^{2}d\xi & \lesssim \frac{N^{4}}{\langle N\rangle ^{4\lambda }\langle \eta
\rangle ^{4\lambda }}\int_{|\xi |\backsim N}\frac{|\xi |^{2}\sin \theta }{%
(\theta ^{2}+d^{2})^{2}}d\xi d\theta \\
& =\frac{N^{6}}{\langle N\rangle ^{4\lambda }\langle \eta \rangle ^{4\lambda
}}\int_{|\xi |\sim N}\left\{ \int_{\theta \leq d}\frac{\theta d\theta }{d^{4}%
}+\int_{\theta \geq d}\frac{d\theta }{\theta ^{3}}\right\} d|\xi | \\
& \lesssim \frac{N^{7}}{\langle N\rangle ^{4\lambda }\langle \eta \rangle
^{4\lambda }}\frac{1}{d^{2}}\lesssim \frac{N^{5}}{\langle N\rangle
^{4\lambda -2}\langle \eta \rangle ^{4\lambda -2}}.
\end{split}
\label{2M2}
\end{equation}%
Next, since $\beta =\theta +\gamma $ and $\gamma ,\beta \lesssim \theta ,$
and $d\backsim \frac{N}{\langle \eta \rangle \langle N\rangle },$ we have 
\begin{equation}
\begin{split}
& \int |\Delta _{\xi }\{\mathfrak{M}_{1}\varphi (\frac{\xi }{N})\chi (\frac{%
2|\xi |}{|\eta |})\}|^{2}d\xi \\
& \lesssim \frac{N^{6}}{\langle N\rangle ^{4\lambda }\langle \eta \rangle
^{4\lambda }}\int_{|\xi |\backsim N}\left\{ \frac{1}{\{\theta
^{2}+d^{2}\}^{4}N^{4}}+\frac{1}{\langle N\rangle ^{12}\{\theta
^{2}+d^{2}\}^{6}}+\frac{\sin ^{4}\beta }{N^{4}\{\theta ^{2}+d^{2}\}^{6}}%
\right\} \theta d\theta d|\xi | \\
& \lesssim \frac{N^{6}}{\langle N\rangle ^{4\lambda }\langle \eta \rangle
^{4\lambda }}\times \int_{|\xi |\backsim N}\Big(\left\{ \int_{\theta \leq d}%
\frac{\theta d\theta }{d^{8}N^{4}}+\int_{\theta \geq d}\frac{d\theta }{%
\theta ^{7}N^{4}}\right\} \\
& +\frac{1}{\langle N\rangle ^{12}}\left\{ \int_{\theta \leq d}\frac{\theta
d\theta }{d^{12}}+\int_{\theta \geq d}\frac{d\theta }{\theta ^{11}}\right\} +%
\frac{1}{N^{4}}\left\{ \int_{\theta \leq d}\frac{\theta ^{5}d\theta }{d^{12}}%
+\int_{\theta \geq d}\frac{d\theta }{\theta ^{7}}\right\} \Big)d|\xi | \\
& \lesssim \frac{N^{3}}{\langle N\rangle ^{4\lambda }\langle \eta \rangle
^{4\lambda }d^{6}}+\frac{N^{7}}{\langle N\rangle ^{4\lambda }\langle \eta
\rangle ^{4\lambda }\langle N\rangle ^{12}d^{10}}+\frac{N^{3}}{\langle
N\rangle ^{4\lambda }\langle \eta \rangle ^{4\lambda }d^{6}} \\
& \lesssim \frac{1}{\langle N\rangle ^{4\lambda -6}\langle \eta \rangle
^{4\lambda -6}N^{3}},
\end{split}
\label{2MH2}
\end{equation}%
where we have used the fact $\frac{1+|\eta |^{2}}{1+N^{2}}\lesssim 1.$ By
interpolation between \eqref{2M2} and \eqref{2MH2}, 
\begin{eqnarray*}
\Vert \mathfrak{M}_{1}\varphi (\frac{\xi }{N})\chi (\frac{2|\xi |}{|\eta |}%
)\Vert _{\dot{H}_{\xi }^{\sigma }} &\lesssim &\Vert \mathfrak{M}_{1}\varphi (%
\frac{\xi }{N})\chi (\frac{2|\xi |}{|\eta |})\Vert _{L^{2}}^{1-\frac{\sigma 
}{2}}\Vert \Delta _{\xi }\{\mathfrak{M}_{1}\varphi (\frac{\xi }{N})\chi (%
\frac{2|\xi |}{|\eta |})\}\Vert _{L^{2}}^{\frac{\sigma }{2}} \\
&\lesssim &\frac{N^{\frac{5}{2}(1-\frac{\sigma }{2})}}{\langle N\rangle
^{2\lambda -1}\langle \eta \rangle ^{2\lambda -1}}\frac{\langle \eta \rangle
^{\sigma }\langle N\rangle ^{\sigma }}{N^{\frac{3}{4}\sigma }} \\
&\lesssim &\frac{\langle \eta \rangle ^{\sigma }\langle N\rangle ^{\sigma
}N^{\frac{5}{2}-2\sigma }}{\langle N\rangle ^{2\lambda -1}\langle \eta
\rangle ^{2\lambda -1}}
\end{eqnarray*}%
as $|\eta |\lesssim N.$ By taking $\sigma =\frac{5}{4}-\varepsilon ,$ this
is summable for $N$ for $4\lambda -2>\frac{5}{2}.$ This concludes %
\eqref{mglobal}. On the other hand, in $\Omega $ (see \eqref{RegionOmega})
we have $|\xi |\simeq |\xi -\eta |\geq 1$ so that $N\gtrsim 1$ and we deduce
that this is summable for $N\geq 1$, $\sigma =3/2-\varepsilon $ when $%
\lambda >1.$ This concludes \eqref{mlocal}.

\medskip

We now turn to the $\eta$ derivatives. Using again Lemma \ref{EstimDerOfPhi}
, we have that 
\begin{equation*}
|\nabla _{\eta }\Phi _{1}(\xi ,\eta )|\lesssim \frac{|\xi |}{\langle \eta
\rangle ^{2}\langle \xi \rangle }+|\sin \beta |,\hskip.1cm\hbox{and}\hskip%
.1cm|\Delta _{\eta }\Phi _{1}(\xi ,\eta )|\lesssim \frac{1}{|\eta |}
\end{equation*}
and therefore 
\begin{eqnarray*}
|\nabla _{\eta }\{\frac{1}{\Phi _{1}}\}| &=&\left\vert -\frac{\nabla _{\eta
}\Phi _{1}}{\Phi _{1}^{2}}\right\vert \lesssim \frac{|\xi |}{|\eta
|^{2}\{\theta ^{2}+d^{2}\}^{2}\langle \xi \rangle \langle \eta \rangle ^{2}}+%
\frac{\sin \beta }{|\eta |^{2}\{\theta ^{2}+d^{2}\}^{2}}, \\
|\Delta _{\eta }\{\frac{1}{\Phi _{1}}\}| &=&\left\vert -\frac{\Delta _{\eta
}\Phi _{1}}{\Phi _{1}^{2}}+2\frac{|\nabla _{\eta }\Phi _{1}|^{2}}{\Phi
_{1}^{3}}\right\vert \\
&\lesssim &\frac{1}{|\eta |^{3}\{\theta ^{2}+d^{2}\}^{2}}+\frac{1}{|\eta
|^{3}\{\theta ^{2}+d^{2}\}^{3}}\left\{ \frac{|\xi |^{2}}{\langle \xi \rangle
^{2}\langle \eta \rangle ^{4}}+\sin ^{2}\beta \right\} .
\end{eqnarray*}%
Define $g$ by 
\begin{equation*}
g=\frac{|\xi ||\xi -\eta ||\eta |}{\langle \xi -\eta \rangle ^{2\lambda
}\langle \eta \rangle ^{2\lambda }}\varphi (\frac{\eta }{M})\chi (\frac{%
2|\eta |}{|\xi |}).
\end{equation*}%
Since $|\eta |\backsim M$ and $|\eta |\lesssim |\xi |,$ direct computation
yields 
\begin{eqnarray*}
|\nabla _{\eta }g| &\lesssim &\frac{1}{|\eta |}\frac{|\xi ||\xi -\eta ||\eta
|}{\langle \xi -\eta \rangle ^{2\lambda }\langle \eta \rangle ^{2\lambda }}%
\mathbf{1}_{|\eta |\backsim M,|\eta |\leq |\xi |} \\
|\partial _{\eta }^{2}g| &\lesssim &\frac{1}{|\eta |^{2}}\frac{|\xi ||\xi
-\eta ||\eta |}{\langle \xi -\eta \rangle ^{2\lambda }\langle \eta \rangle
^{2\lambda }}\mathbf{1}_{|\eta |\backsim M,|\eta |\leq |\xi |}.
\end{eqnarray*}%
Hence, since $\sin \beta \lesssim \sin \theta $, 
\begin{eqnarray*}
|\Delta _{\eta }\{\mathfrak{M}_{1}\varphi (\frac{\eta }{M})\chi (\frac{2|\xi
|}{|\eta |})\}| &=&|\Delta _{\eta }\{\frac{1}{\Phi _{1}}\}g+2\nabla _{\eta
}\{\frac{1}{\Phi _{1}}\}\cdot \nabla _{\eta }g+\frac{1}{\Phi _{1}}\Delta
_{\eta }g| \\
&\lesssim &\frac{|\xi |^{2}}{M^{2}\{\theta ^{2}+d^{2}\}^{2}\langle M\rangle
^{2\lambda }\langle \xi \rangle ^{2\lambda }} \\
&&+\frac{|\xi |^{2}}{M^{2}\{\theta ^{2}+d^{2}\}^{3}\langle M\rangle
^{2\lambda }\langle \xi \rangle ^{2\lambda }}\left\{ \frac{|\xi |^{2}}{%
\langle \xi \rangle ^{2}\langle \eta \rangle ^{4}}+\sin ^{2}\theta \right\} .
\end{eqnarray*}%
By using $\frac{\xi}{|\xi|}$ as the north pole, and $d\backsim \frac{|\xi |}{%
\langle M\rangle \langle \xi \rangle },$ we thus compute that 
\begin{equation}  \label{2M2eta}
\begin{split}
\int_{|\eta |\sim M}|\mathfrak{M}_{1}\varphi (\frac{\eta }{M})\chi (\frac{%
2|\xi |}{|\eta |})|^{2}d\eta &\lesssim \frac{|\xi |^{4}}{\langle M\rangle
^{4\lambda }\langle \xi \rangle ^{4\lambda }}\int_{|\eta |\backsim M}\frac{%
\sin \theta }{(\theta ^{2}+d^{2})^{2}}d\eta d\theta \\
&\lesssim \frac{|\xi |^{4}M^{2}}{\langle M\rangle ^{4\lambda }\langle \xi
\rangle ^{4\lambda }}\int_{|\eta |\backsim M}\left\{ \int_{\theta \leq d}%
\frac{\theta d\theta }{d^{4}}+\int_{\theta \geq d}\frac{d\theta }{\theta ^{3}%
}\right\} d|\eta | \\
&\lesssim \frac{|\xi |^{4}M^{3}}{\langle M\rangle ^{4\lambda }\langle \xi
\rangle ^{4\lambda }}\frac{1}{d^{2}}\lesssim \frac{|\xi |^{2}M^{3}}{\langle
M\rangle ^{4\lambda -2}\langle \xi \rangle ^{4\lambda -2}}.
\end{split}%
\end{equation}
Next, since $\beta =\theta +\gamma $ and $\gamma ,\beta \lesssim \theta $,
and we have that 
\begin{equation}  \label{2MH2eta}
\begin{split}
&\int_{|\eta |\sim M}|\Delta _{\eta }\{\mathfrak{M}_{1}\varphi (\frac{\eta }{%
M})\chi (\frac{2|\xi |}{|\eta |})\}|^{2}d\eta \\
&\lesssim \frac{1}{M^{2}\langle M\rangle ^{4\lambda }\langle \xi \rangle
^{4\lambda }} \\
&\times \int_{\substack{  \\ |\eta |\backsim M}}\left\{ \frac{|\xi |^{4}}{%
\{\theta ^{2}+d^{2}\}^{4}}+\frac{|\xi |^{8}}{\{\theta
^{2}+d^{2}\}^{6}\langle \xi \rangle ^{4}\langle \eta \rangle ^{8}}+\frac{%
|\xi |^{4}\sin ^{4}\theta }{\{\theta ^{2}+d^{2}\}^{6}}\right\} \theta
d\theta d|\eta | \\
&\lesssim \frac{|\xi |^{4}}{M^{2}\langle M\rangle ^{4\lambda }\langle \xi
\rangle ^{4\lambda }}\int_{|\eta |\backsim M}\Big(\left\{ \int_{\theta
\lesssim d}\frac{\theta d\theta }{d^{8}}+\int_{\theta \geq d}\frac{d\theta }{%
\theta ^{7}}\right\} \\
&+\frac{|\xi |^{8}}{\langle \xi \rangle ^{4}\langle M\rangle ^{8}}\left\{
\int_{\theta \leq d}\frac{\theta d\theta }{d^{12}}+\int_{\theta \geq d}\frac{%
d\theta }{\theta ^{11}}\right\} +\left\{ \int_{\theta \lesssim d}\frac{%
\theta ^{5}d\theta }{d^{12}}+\int_{\theta \geq d}\frac{d\theta }{\theta ^{7}}%
\right\} \Big)d|\eta | \\
&\lesssim \frac{|\xi |^{4}}{M\langle M\rangle ^{4\lambda }\langle \xi
\rangle ^{4\lambda }d^{6}}+\frac{|\xi |^{8}}{M\langle M\rangle ^{4\lambda
+8}\langle \xi \rangle ^{4\lambda +4}d^{10}}+\frac{|\xi |^{4}}{M\langle
M\rangle ^{4\lambda }\langle \xi \rangle ^{4\lambda }d^{6}} \\
&\lesssim \frac{1}{M\langle M\rangle ^{4\lambda -6}\langle \xi \rangle
^{4\lambda -6}|\xi |^{2}}.
\end{split}%
\end{equation}
Interpolating between \eqref{2M2eta} and \eqref{2MH2eta}, we obtain 
\begin{eqnarray*}
\Vert \mathfrak{M}_{1}\varphi (\frac{\eta }{M})\chi (\frac{2|\xi |}{|\eta |}%
)\Vert _{\dot{H}_{\eta }^{\sigma }} &\lesssim &\Vert \mathfrak{M}_{1}\varphi
(\frac{\eta }{M})\chi (\frac{2|\xi |}{|\eta |})\Vert _{L^{2}}^{1-\frac{%
\sigma }{2}}\Vert \Delta _{\eta }\{\mathfrak{M}_{1}\varphi (\frac{\eta }{M}%
)\chi (\frac{2|\xi |}{|\eta |})\}\Vert _{L^{2}}^{\frac{\sigma }{2}} \\
&\lesssim &\left\{ \frac{|\xi |^{2}M^{3}}{\langle M\rangle ^{4\lambda
-2}\langle \xi \rangle ^{4\lambda -2}}\right\} ^{\frac{1}{2}-\frac{\sigma }{4%
}}\left\{ \frac{1}{M\langle M\rangle ^{4\lambda -6}\langle \xi \rangle
^{4\lambda -6}|\xi |^{2}}\right\} ^{\frac{\sigma }{4}} \\
&\lesssim &\frac{M^{\frac{3}{2}-\sigma }}{\langle M\rangle ^{2\lambda
-1-\sigma }\langle \xi \rangle ^{2\lambda -1-\sigma }|\xi |^{\sigma -1}}
\end{eqnarray*}%
By taking $\sigma =\frac{5}{4}-\varepsilon ,$ this is summable in $M$ if $%
2\lambda -1-\sigma >0$ and we conclude \eqref{mglobal}. On the other hand,
in $\Omega ,$ we have $|\xi |\geq 1$ so that by taking $\sigma =\frac{3}{2}%
-\varepsilon ,$ this is summable for $M\geq 1$ if $\lambda >1$ and if $f=%
\frac{1}{\langle \xi -\eta \rangle ^{\frac{1}{2}}}$ or $\frac{1}{\langle
\eta \rangle ^{\frac{1}{2}}}$.

\medskip

\textbf{Case 3 Region }$\Omega _{3}=\{\frac{1}{3}<\frac{|\xi |}{|\eta |}
<3\}. $

In this region, we have $|\xi -\eta |\leq 4\min (|\xi |,|\eta |)$, $|\xi
|\simeq |\eta |$ are of the order of the longest side and $\sin \gamma
\simeq \sin \beta $. Therefore 
\begin{equation}
|\Phi _{1}|\geq |\xi -\eta |(\gamma ^{2}+\beta ^{2}+\frac{|\xi |^{2}+|\eta
|^{2}}{(1+|\xi |^{2}+|\eta |^{2})\langle \xi -\eta \rangle ^{2}})\equiv |\xi
-\eta |(\gamma ^{2}+\beta ^{2}+d_{1}^{2}).  \label{d1}
\end{equation}%
The above lower bound is trivial if $|\xi |$ is not the largest. If $|\xi |$
is the largest and $|\xi -\eta |$ is not the smallest, then $\xi ,\eta ,\xi
-\eta $ are all comparable so that $\gamma \simeq \theta \simeq \pi -\beta $
and from \eqref{EstimOnPhi}, 
\begin{equation*}
|\Phi _{1}|\gtrsim \theta ^{2}|\eta |+\frac{|\eta |}{1+|\eta |^{2}}\gtrsim
|\xi -\eta |(\gamma ^{2}+\beta ^{2}+d_{1}^{2}).
\end{equation*}%
Finally, when $|\xi -\eta |$ is the smallest, this follows from %
\eqref{EstimOnPhi}. Moreover, from \eqref{phixi} and \eqref{phieta}, 
\begin{equation}
|\partial _{\xi ,\eta }\Phi _{1}(\xi ,\eta )|\lesssim \frac{|\eta |}{\langle
\eta \rangle \langle \xi -\eta \rangle ^{2}}+|\sin \gamma |,\text{ \ \ \ \ \ 
}|\Delta _{\xi ,\eta }\Phi _{1}(\xi ,\eta )|\lesssim \frac{1}{|\xi -\eta |}
\end{equation}%
and therefore, 
\begin{eqnarray*}
|\nabla _{\xi ,\eta }\{\frac{1}{\Phi _{1}}\}| &=&\left\vert -\frac{\nabla
_{\xi }\Phi _{1}}{\Phi _{1}^{2}}\right\vert \\
&\lesssim &\frac{|\eta |}{|\xi -\eta |^{2}(\beta ^{2}+\gamma
^{2}+d_{1}^{2})^{2}\langle \eta \rangle \langle \xi -\eta \rangle ^{2}}+%
\frac{\sin \gamma }{|\xi -\eta |^{2}(\beta ^{2}+\gamma ^{2}+d_{1}^{2})^{2}}
\\
|\Delta _{\xi ,\eta }\{\frac{1}{\Phi _{1}}\}| &=&\left\vert -\frac{\Delta
_{\xi }\Phi _{1}}{\Phi _{1}^{2}}+2\frac{|\nabla _{\xi }\Phi _{1}|^{2}}{\Phi
_{1}^{3}}\right\vert \\
&\lesssim &\frac{1}{|\xi -\eta |^{3}(\beta ^{2}+\gamma ^{2}+d_{1}^{2})^{2}}+%
\frac{|\eta |^{2}+|\xi |^{2}}{|\xi -\eta |^{3}(\beta ^{2}+\gamma
^{2}+d_{1}^{2})^{3}\langle \eta \rangle ^{2}\langle \xi -\eta \rangle ^{4}}
\\
&&+\frac{\sin ^{2}\gamma +\sin ^{2}\beta }{|\xi -\eta |^{3}(\beta
^{2}+\gamma ^{2}+d_{1}^{2})^{3}}.
\end{eqnarray*}%
For fixed $\eta $ and a dyadic number $N$, denote 
\begin{equation*}
g=\frac{|\xi ||\xi -\eta ||\eta |}{\langle \xi -\eta \rangle ^{2\lambda
}\langle \eta \rangle ^{2\lambda }}\varphi (\frac{\xi -\eta }{N})\varphi (%
\sqrt{\frac{|\eta |}{|\xi |}})
\end{equation*}%
As before, direct computation yields 
\begin{eqnarray*}
|\nabla _{\xi ,\eta }g| &\lesssim &\frac{1}{|\xi -\eta |}\frac{|\xi ||\xi
-\eta ||\eta |}{\langle \xi -\eta \rangle ^{2\lambda }\langle \eta \rangle
^{2\lambda }}\mathbf{1}_{|\xi -\eta |\backsim N,|\xi |\backsim |\eta |} \\
|\partial _{\xi ,\eta }^{2}g| &\lesssim &\frac{1}{|\xi -\eta |^{2}}\frac{%
|\xi ||\xi -\eta ||\eta |}{\langle \xi -\eta \rangle ^{2\lambda }\langle
\eta \rangle ^{2\lambda }}\mathbf{1}_{|\xi -\eta |\backsim N,|\xi |\backsim
|\eta |}.
\end{eqnarray*}%
Therefore, 
\begin{equation*}
|\mathfrak{M}_{1}\varphi (\frac{\xi -\eta }{N})\varphi (\sqrt{\frac{|\eta |}{%
|\xi |}})|\lesssim \frac{|\eta |^{2}\mathbf{1}_{|\xi -\eta |\backsim N,|\xi
|\backsim |\eta |}}{\langle N\rangle ^{2\lambda }\langle \eta \rangle
^{2\lambda }(\beta ^{2}+d_{1}^{2})},
\end{equation*}%
and 
\begin{eqnarray*}
&&|\Delta _{\xi }\{\mathfrak{M}_{1}\varphi (\frac{\xi -\eta }{N})\varphi (%
\sqrt{\frac{|\eta |}{|\xi |}})\}| \\
&\lesssim &\left\{ \frac{|\eta |^{2}}{(\beta ^{2}+d_{1}^{2})^{2}}+\frac{%
|\eta |^{4}}{(\beta ^{2}+d_{1}^{2})^{3}\langle \eta \rangle ^{2}\langle \xi
-\eta \rangle ^{4}}+\frac{|\eta |^{2}\sin ^{2}\beta }{(\beta
^{2}+d_{1}^{2})^{3}}\right\} \frac{\mathbf{1}_{|\xi -\eta |\backsim N,|\xi
|\backsim |\eta |}}{N^{2}\langle N\rangle ^{2\lambda }\langle \eta \rangle
^{2\lambda }}.
\end{eqnarray*}%
By using $-\frac{\eta }{|\eta |}$ as the north pole, and $d_{1}\backsim 
\frac{|\eta |}{\langle \eta \rangle \langle N\rangle },$ we thus compute
from \eqref{d1}: 
\begin{equation}  \label{3M2}
\begin{split}
&\int |\mathfrak{M}_{1}\varphi (\frac{\xi -\eta }{N})\varphi (\sqrt{\frac{%
|\eta |}{|\xi |}})|^{2}d\xi \\
&\lesssim \frac{|\eta |^{4}N^{2}}{\langle N\rangle ^{4\lambda }\langle \eta
\rangle ^{4\lambda }}\int_{|\xi -\eta |\backsim N}\frac{\sin \beta }{(\beta
^{2}+d_{1}^{2})^{2}}d|\xi |d\beta \\
&=\frac{|\eta |^{4}N^{3}}{\langle N\rangle ^{4\lambda }\langle \eta \rangle
^{4\lambda }}\int_{|\xi -\eta |\sim N}\left\{ \int_{\beta \leq d_{1}}\frac{%
\beta d\beta }{d_{1}^{4}}+\int_{\beta \geq d}\frac{d\beta }{\beta ^{3}}%
\right\} d|\xi | \\
&\lesssim \frac{|\eta |^{4}N^{3}}{\langle N\rangle ^{4\lambda }\langle \eta
\rangle ^{4\lambda }}\frac{1}{d_{1}^{2}}\lesssim \frac{|\eta |^{2}N^{3}}{%
\langle N\rangle ^{4\lambda -2}\langle \eta \rangle ^{4\lambda -2}}.
\end{split}%
\end{equation}
Next, we have that 
\begin{equation}  \label{3MH2}
\begin{split}
&\int_{|\xi -\eta |\sim N}|\Delta _{\xi }\{\mathfrak{M}_{1}\varphi (\frac{%
\xi -\eta }{N})\varphi (\sqrt{\frac{|\eta |}{|\xi |}})\}|^{2}d\xi \\
&\lesssim \frac{|\eta |^{4}}{\langle N\rangle ^{4\lambda }\langle \eta
\rangle ^{4\lambda }N^{4}}\int_{|\xi -\eta |\backsim N}\left\{ \frac{1}{%
(\beta ^{2}+d_{1}^{2})^{4}}+\frac{|\eta |^{4}}{(\beta
^{2}+d_{1}^{2})^{6}\langle \eta \rangle ^{4}\langle N\rangle ^{8}}+\frac{%
\beta ^{4}}{(\beta ^{2}+d_{1}^{2})^{6}}\right\} d(\xi -\eta ) \\
&\lesssim \frac{|\eta |^{4}N^{2}}{\langle N\rangle ^{4\lambda }\langle \eta
\rangle ^{4\lambda }N^{4}}\int_{|\xi -\eta |\backsim N}d|\xi -\eta |\Big(%
\left\{ \int_{\beta \leq d_{1}}\frac{\beta d\beta }{d_{1}^{8}}+\int_{\beta
\geq d_{1}}\frac{d\beta }{\beta ^{7}}\right\} \\
&+\frac{|\eta |^{4}}{\langle \eta \rangle ^{4}\langle N\rangle ^{8}}\left\{
\int_{\beta \leq d_{1}}\frac{\beta d\beta }{d_{1}^{12}}+\int_{\beta \geq
d_{1}}\frac{d\beta }{\beta ^{11}}\right\} +\left\{ \int_{\beta \leq d_{1}}%
\frac{\beta ^{5}d\beta }{d_{1}^{12}}+\int_{\beta \geq d_{1}}\frac{d\beta }{%
\beta ^{7}}\right\} \Big) \\
&\lesssim \frac{|\eta |^{4}}{\langle N\rangle ^{4\lambda }\langle \eta
\rangle ^{4\lambda }Nd^{6}}+\frac{|\eta |^{8}}{\langle N\rangle ^{4\lambda
+8}\langle \eta \rangle ^{4\lambda +4}Nd_{1}^{10}}+\frac{|\eta |^{4}}{%
\langle N\rangle ^{4\lambda }N\langle \eta \rangle ^{4\lambda }d_{1}^{6}} \\
&\lesssim \frac{1}{\langle N\rangle ^{4\lambda -6}\langle \eta \rangle
^{4\lambda -6}|\eta |^{2}N}.
\end{split}%
\end{equation}
Interpolating between \eqref{3M2} and \eqref{3MH2}, we have 
\begin{eqnarray*}
&&||\mathfrak{M}_{1}\varphi (\frac{\xi -\eta }{N})\varphi (\sqrt{\frac{|\eta
|}{|\xi |}})||_{\dot{H}_{\xi }^{\sigma }} \\
&\lesssim &||\mathfrak{M}_{1}\varphi (\frac{\xi -\eta }{N})\varphi (\sqrt{%
\frac{|\eta |}{|\xi |}})||_{L^{2}}^{1-\frac{\sigma }{2}}||\Delta _{\xi }\{%
\mathfrak{M}_{1}\varphi (\frac{\xi -\eta }{N})\varphi (\sqrt{\frac{|\eta |}{%
|\xi |}})\}||_{L^{2}}^{\frac{\sigma }{2}} \\
&\lesssim &\left\{ \frac{|\eta |^{2}N^{3}}{\langle N\rangle ^{4\lambda
-2}\langle \eta \rangle ^{4\lambda -2}}\right\} ^{\frac{1}{2}-\frac{\sigma }{%
4}}\left\{ \frac{1}{\langle N\rangle ^{4\lambda -6}\langle \eta \rangle
^{4\lambda -6}|\eta |^{2}N}\right\} ^{\frac{\sigma }{4}} \\
&\lesssim &\frac{N^{\frac{3}{2}-\sigma }}{\langle N\rangle ^{2\lambda
-1-\sigma }\langle \eta \rangle ^{2\lambda -1-\sigma }|\eta |^{\sigma -1}}
\end{eqnarray*}%
as $N\lesssim |\eta |.$ By taking $\sigma =\frac{5}{4}-\varepsilon ,$ this
is summable for $N$ when $4\lambda -2>\frac{5}{2}$ hence we deduce (\ref%
{mglobal})$.$ On the other hand, in $\Omega ,$ we know that $|\eta |\simeq
|\xi |\geq 1.$ Hence for $f=\frac{1}{\langle \xi -\eta \rangle ^{\frac{1}{2}}%
}$ or $\frac{1}{\langle \eta \rangle ^{\frac{1}{2}}}$, we can take $\sigma =%
\frac{3}{2}-\varepsilon $ and still get a convergent series. \eqref{mlocal}
therefore follows. The $\eta $ derivatives can be controlled similarly since
we had the same control.
\end{proof}

%%%%%%%%%%%%%%%%%%%%%%%%%%%%%%%%%%%%%%%%%%%%%%%%%%%%%%%%%%%%%%%%%%%%%%%%%%%%%%%%%%%%%%%%%%%%%%%%%%%%%%%%%%%%%%%%%%%%%%%%%%%%%%%%%%%%%%%%%%%%%%%%%%%%%%%%%%%%%%%%%%%%%%%%%%%%%%%%%%%%%%%%%%%%%%%%%%%%%%%%%%%%%%%%%%%%%%%%%%%%%%%%%%%%%%%%%%%%%%%%%%%%%%%%%%%%%%%%%%%%%%%%%%%%%%%%%%%%%%%%%%%%%%%%%%%%%%%%%%%%%%%%%%%%%%%%%%%%%%%%%%%%%%%

\section{The $L^{10}$ Bound and end of the proof}

\label{SecL10}

\subsection{Estimating the $L^{10}$ bound}

Using the results of Section \ref{SecMult}, we can now estimate the last
part of the $X$ norm.

\begin{proposition}
\label{EstimL10NormForAlphaProp} Let $\alpha$ be a solution of \eqref{alpha}%
, then 
\begin{equation}  \label{EstimL10NormForAlpha}
\sup_{t\geq 0}(1+t)^{\frac{16}{15}}\Vert \alpha(t)\Vert _{W^{k,10}}\lesssim
\Vert \alpha _{0}\Vert _{Y}+\Vert \alpha \Vert _{X}^{2}.
\end{equation}
\end{proposition}

\begin{proof}
We use Theorem \ref{CorGNT} and Proposition \ref{EstimGenPhase} to control
the nonlinear terms appearing in \eqref{alpha}. Our strategy is first to
establish the Proposition for $\Phi _{1},$ and then we use symmetry to
conclude all the other cases. We first deal with the cubic terms as follows.
We let 
\begin{equation*}
A(\xi-\eta)=\frac{n_2(\xi-\eta)}{\vert\xi-\eta\vert}\langle\xi-\eta%
\rangle^{2\lambda}\hskip.1cm\hbox{and}\hskip.1cm B(\eta)=\frac{n_3(\eta)}{%
\vert\eta\vert}\langle\eta\rangle^{2\lambda}
\end{equation*}
We first apply Proposition \ref{PropLinDecay} to get that, for a typical
term, 
\begin{equation*}
\begin{split}
& \Vert (1-\Delta)^{\frac{k}{2}}\mathcal{F}^{-1}\left\{%
\int_{0}^{t}e^{i(t-s)p(|\xi |)}\frac{im^1(\xi ,\eta )}{\Phi_{1}}\hat{\alpha}%
(s,\xi -\eta ) \hat{h}(\alpha)(s,\eta)ds\right\}\Vert _{L^{10}} \\
\lesssim &\int_{0}^{t}\frac{1}{(1+t-s)^{\frac{16}{15}}}\Vert \mathcal{F}
_{\xi }^{-1}\left\{n_{1}(\xi )\langle \xi \rangle ^{k}\int_{\mathbb{R}^{3}}%
\frac{|\xi |}{\Phi _{1}}n_{2}(\xi -\eta)\hat{\alpha}(\xi -\eta)n_{3}(\eta )%
\hat{h}(\eta )d\eta\right\} \Vert _{W^{\frac{12}{5},\frac{10}{9}}}ds \\
\lesssim &\int_{0}^{t}\frac{1}{(1+t-s)^{\frac{16}{15}}}\Vert \mathcal{F}
_{\xi }^{-1}\left\{\langle \xi \rangle ^{k+\frac{12}{5}}\int_{\mathbb{R}%
^{3}} \mathfrak{M}_{1}A(\xi-\eta)\hat{\alpha}(\xi -\eta )B(\eta)\hat{h}(\eta
)\right\}\Vert _{L^{\frac{10}{9}}}ds \\
\lesssim &\int_{0}^{t}\frac{1}{(1+t-s)^{\frac{16}{15}}}\Big(\Vert
A(\vert\nabla\vert)\alpha \Vert_{H^{k+\frac{12}{5}}}\Vert
B(\vert\nabla\vert)\beta \Vert_{L^{l_2}}+\Vert
A(\vert\nabla\vert)\alpha\Vert_{L^{l_2}}\Vert B(\vert\nabla\vert)\beta
\Vert_{H^{k+\frac{12}{5}}}\Big) ds \\
\lesssim &\int_{0}^{t}\frac{1}{(1+t-s)^{\frac{16}{15}}}\left( \Vert \alpha
\Vert _{H^{-1}\cap H^{k+2\lambda+\frac{7}{5}}}\Vert |\nabla|^{-1}\beta \Vert
_{H^{3}}+\Vert \alpha \Vert _{H^{-1}\cap H^{2}}\Vert
\vert\nabla\vert^{-1}\beta \Vert _{H^{k+2\lambda+\frac{12}{5}}}\right) ds \\
\lesssim &(1+t)^{-\frac{16}{15}}\Vert \alpha \Vert _{X}^{3}
\end{split}%
\end{equation*}
since $k\ge 2\lambda+\frac{7}{5}$. Here we have applied Lemma \ref{infinity}
around $s=\frac{5}{ 4}-\varepsilon,$ Proposition \ref{EstimGenPhase} and
Theorem \ref{CorGNT} with $l_{1}=10$, $b=\infty$, and $l_{2}= \frac{60}{%
29+20\varepsilon} >2$. To finish the analysis of the cubic term, we also
need to control the cubic term pre-normal form in \eqref{alpha}. We use the
fact that $e^{itp(\vert\nabla\vert)}$ is a unitary operator on $H^{k}$ and %
\eqref{PropertiesOfN} to get that 
\begin{equation*}
\begin{split}
\Vert \mathcal{F}^{-1}\int_0^te^{i(t-s)p(\xi)}\hat{\mathcal{N}}%
_1(\alpha)(s)ds\Vert_{W^{k,10}}&\lesssim \int_0^t\Vert
e^{i(t-s)p(\vert\nabla\vert)}\mathcal{N}_1(\alpha)(s)\Vert_{H^{k+2}}ds \\
&\lesssim \int_0^t\Vert\vert\nabla\vert^{-1}\mathcal{N}_1(\alpha)(s)%
\Vert_{H^{k+3}}ds \\
&\lesssim \Vert\alpha\Vert_X^2
\end{split}%
\end{equation*}

\medskip

To estimate the integrated term $\mathfrak{B}$ in \eqref{alpha}, we need to
separate the regions. First we control the integrated part when all the
terms are small, 
\begin{equation*}
M=\max (|\xi |,|\xi -\eta |,|\eta |)<3.
\end{equation*}
To do this, we first note that Sobolev's embedding $\dot{H}^{\frac{6}{5}
}\subset L^{10},$ and the fact that for bounded $\xi $, $|\xi |^{\frac{6}{5}
}\lesssim |\xi |\lesssim |\xi -\eta |+|\eta |$ to get\footnote{%
Here we forget the difference between $n_2$ and $n_3$ and treat the terms as
symmetric.} 
\begin{equation*}
\begin{split}
& \Vert \mathfrak{B}(\alpha (t),\alpha (t))\Vert _{L^{10}} \\
& \lesssim \Vert \mathcal{F}_{\xi }^{-1}\left\{|\xi |^{\frac{6}{5}}\int_{%
\mathbb{R}^{3}}\frac{|\xi ||\xi -\eta |||\eta |}{i\Phi _{1}}\frac{m}{|\xi |}%
\frac{\hat{\alpha}(t,\xi -\eta )}{|\xi -\eta |}\frac{\hat{\alpha}(t,\eta )}{%
|\eta |}d\eta\right\} \Vert _{L^{2}} \\
& \lesssim \Vert \int_{\mathbb{R}^{3}}\mathfrak{M}_{1}n_{1}(\xi )n_{2}(\xi
-\eta )n_{3}(\eta )\langle \xi -\eta \rangle ^{2\lambda }\hat{\alpha}(t,\xi
-\eta )\frac{\langle \eta \rangle ^{2\lambda }\hat{\alpha}(t,\eta )}{|\eta |}%
d\eta \Vert _{L^{2}} \\
& \lesssim \Vert n_{2}(|\nabla |)\alpha \Vert _{L^{10}}\Vert |\nabla
|^{-1}n_{3}(|\nabla |)\alpha \Vert _{L^{2}}\lesssim \Vert \alpha \Vert
_{X}\Vert \alpha \Vert _{L^{10}}
\end{split}%
\end{equation*}
where we have applied Proposition \ref{EstimGenPhase} for $\mathfrak{M}_{1}$
for $s=\frac{5}{4} -\varepsilon$, Lemma \ref{infinity}, and Lemma \ref%
{LemGNT} with $l_{1}=2$, $b=\infty $, $l_{2}=10$ and $l_{3}=\frac{60}{%
29+20\varepsilon}>2$.

Next we deal with the case when one of the frequencies is large $M>1.$ This
can happen in two cases. First, if $|\xi |\leq 1,M>1$. In this case, we have 
$|\eta |\simeq |\xi -\eta |\geq 1.$ We bound the $L^{10}$ norm by the $L^{2}$
norm via Sobolev's inequality (with bounded $\xi $) to get 
\begin{equation*}
\begin{split}
\Vert \mathfrak{B}\Vert _{L^{10}}& \lesssim \Vert \mathcal{F}_{\xi
}^{-1}\int_{\mathbb{R}^{3}}\frac{m(\xi ,\eta )}{i\Phi _{1}}\chi \hat{\alpha}%
(t,\xi -\eta )\hat{\alpha}(t,\eta )d\eta \Vert _{L^{10}} \\
& \lesssim \Vert \mathcal{F}_{\xi }^{-1}\int_{\mathbb{R}^{3}}n_{1}(\xi )f 
\mathfrak{M}_{j}n_{2}(\xi -\eta )\frac{\langle \xi -\eta \rangle ^{2\lambda +%
\frac{1}{4}}\hat{\alpha}(t,\xi -\eta )}{|\xi -\eta |}n_{3}(\eta )\frac{%
\langle \eta \rangle ^{2\lambda +\frac{1}{4}}\hat{\alpha}(t,\eta )}{|\eta |}%
d\eta \Vert _{L^{10}} \\
& \lesssim \Vert \mathcal{F}_{\xi }^{-1}\int_{\mathbb{R}^{3}}f\mathfrak{M}%
_{j}n_{2}(\xi -\eta )\frac{\langle \xi -\eta \rangle ^{2\lambda +\frac{1}{4}}%
\hat{\alpha}(t,\xi -\eta )}{|\xi -\eta |}n_{3}(\eta )\frac{\langle \eta
\rangle ^{2\lambda +\frac{1}{4}}\hat{\alpha}(t,\eta )}{|\eta |}d\eta \Vert
_{L^{2}} \\
& \lesssim \Vert \alpha \Vert _{H^{2}}\Vert n_{3}(|\nabla |)\alpha \Vert
_{W^{2,\frac{3}{\varepsilon }}}\lesssim \Vert \alpha \Vert _{X}\Vert \alpha
\Vert _{W^{3,10}}.
\end{split}%
\end{equation*}
We have applied Proposition \ref{EstimGenPhase} with $s=\frac{3}{2}%
-\varepsilon ,$ Lemma \ref{infinity} around $s=\frac{3}{2}-\varepsilon$ and
Lemma \ref{LemGNT} with $s=\frac{3}{2} -\varepsilon$, $b=\infty$, $l_{1}=2$, 
$l_{2}=10$ and $l_{3}=\frac{15}{6+5\varepsilon}>2.$ This concludes the
estimates in the region $\{|\xi |\leq 1\}\cap \{M\geq 1\}$.

The other case is included in the region 
\begin{equation*}
\Omega =\{|\eta |\leq 2|\xi -\eta |,|\xi |\geq 1/2\}\cup \{|\eta |>|\xi
-\eta |,|\xi |\geq 1/2\}
\end{equation*}
and leads to the worst loss in derivatives (whereas the region when all
frequencies are small leads to the loss of smoothness of the multiplier and
and hence to the loss of decay in time). In the case $|\eta |\leq 2|\xi
-\eta |,$ we choose $f=\frac{\chi }{\langle \eta \rangle ^{\frac{1}{2}}}$ in
Proposition \ref{EstimGenPhase}. We apply Lemma \ref{infinity} to deduce
that 
\begin{equation*}
\frac{\vert\xi \vert^{k+\frac{6}{5}}}{\langle \xi -\eta \rangle ^{k+ \frac{6%
}{5}}}\mathfrak{M}_{1}f\in M_{\xi ,\eta }^{\frac{3}{2}-\varepsilon }.
\end{equation*}
Hence 
\begin{equation*}
\begin{split}
& \Vert \vert\nabla\vert^k\mathfrak{B}(\alpha ,\alpha )\Vert _{L^{10}} \\
& \lesssim \Vert \mathcal{F}_{\xi }^{-1}\left\{ |\xi |^{k}n_{1}(\xi )\int_{ 
\mathbb{R}^{3}}f\mathfrak{M}_{j}n_{2}(\xi -\eta)\frac{\langle \xi -\eta
\rangle ^{2\lambda}\hat{\alpha}(\xi -\eta )}{|\xi -\eta|}n_{3}(\eta )\frac{
\langle \eta \rangle ^{2\lambda +\frac{1}{2}}\hat{\alpha}(\eta )}{|\eta |}
d\eta \right\} \Vert _{L^{10}} \\
& \lesssim \Vert \mathcal{F}^{-1}\int_{\mathbb{R}^{3}}[\frac{|\xi |^{k+\frac{%
6}{5}}}{\langle \xi -\eta \rangle ^{k+\frac{6}{5}}}\mathfrak{M}%
_{j}f]n_{2}(\xi -\eta )\frac{\langle \xi -\eta \rangle ^{k+2\lambda +\frac{6%
}{5}}\hat{\alpha}(\xi -\eta )}{|\xi -\eta |}n_{3}(\eta )\frac{\langle \eta
\rangle ^{2\lambda +\frac{1}{2}}}{|\eta |}\hat{\alpha}(\eta )d\eta \Vert
_{L^{2}} \\
& \lesssim \Vert \frac{|\xi |^{k+\frac{6}{5}}}{\langle \xi -\eta \rangle ^{k+%
\frac{6}{5}}}\mathfrak{M}_{j}f\Vert _{\mathcal{M}_{\xi ,\eta }^{\frac{3}{2}%
-\varepsilon }}\Vert |\nabla |^{k+\frac{11}{5}+2\delta }n_{2}(|\nabla
|)\alpha \Vert _{L^{l_{3}}}\Vert \frac{(1-\Delta )^{2}}{|\nabla |}%
n_{3}(|\nabla |)\alpha \Vert _{L^{l_{2}}} \\
& \lesssim \Vert |\nabla |^{k+\frac{11}{5}+2\delta }\alpha \Vert
_{L^{l_{2}}}\Vert \frac{(1-\Delta )^{2}}{|\nabla |}\alpha \Vert _{L^{l_{3}}}.
\end{split}%
\end{equation*}
We have applied Lemma with $s=\frac{3}{2}-\varepsilon ,$ $\frac{1}{l_{2}}=%
\frac{14}{50}+\frac{17}{15}\varepsilon $ and $\frac{1}{l_{3}}=\frac{11}{50}-%
\frac{4}{5}\varepsilon $. Now, using Bernstein estimates, we compute that 
\begin{equation*}
\begin{split}
\Vert P_{\leq 1}\frac{(1-\Delta )^{2}}{|\nabla |}\alpha \Vert _{L^{l_{2}}}&
\lesssim \sum_{N\leq 1}N^{-1}\Vert P_{N}\alpha \Vert _{L^{l_{2}}} \\
& \lesssim \sum_{N\leq 1}N^{-1}N^{3\left( \frac{1}{l_{4}}-\frac{1}{l_{2}}%
\right) }\Vert P_{N}\alpha \Vert _{L^{l_{4}}} \\
& \lesssim \sum_{N\leq 1}N^{\varepsilon }\left( N^{-1}\Vert P_{N}\alpha
\Vert _{L^{2}}\right) ^{1-\sigma }\Vert P_{N}\alpha \Vert _{L^{10}}^{\sigma
}\lesssim \Vert \alpha \Vert _{X}^{1-\sigma }\Vert \alpha \Vert
_{L^{10}}^{\sigma },
\end{split}%
\end{equation*}
for 
\begin{equation*}
\sigma =\frac{5}{11}\left( \frac{3}{b}-2\varepsilon \right) =\frac{3}{10}+ 
\frac{2\varepsilon }{11},\hskip.1cm\hbox{and}\hskip.1cm \frac{1}{l_{4}}=%
\frac{1}{2 }-\frac{2\sigma }{5}=\frac{19}{50}-\frac{4}{55}\varepsilon
\end{equation*}
while for the high frequencies, we have that 
\begin{equation*}
\begin{split}
\Vert P_{\geq 1}\frac{(1-\Delta )^{2}}{|\nabla |}\alpha \Vert _{L^{l_{2}}}&
\lesssim \Vert (1-\Delta )^{\frac{3}{2}}\alpha \Vert _{L^{10}}^{\sigma
}\Vert (1-\Delta )^{\frac{3}{2}}\alpha \Vert _{L^{l_{5}}}^{1-\sigma } \\
& \lesssim \Vert \alpha \Vert _{W^{3,10}}^{\sigma }\Vert \alpha \Vert
_{X}^{1-\sigma }
\end{split}%
\end{equation*}
for $\frac{1}{l_{5}}=(1/4+184\varepsilon )/(7/10-2/11\varepsilon )$.
Independently, we have that 
\begin{equation*}
\begin{split}
\Vert |\nabla |^{k+\frac{11}{5}+2\delta }\alpha \Vert _{L^{l_{3}}}& \lesssim
\Vert (1-\Delta )^{\frac{k}{2}}\alpha \Vert _{L^{10}}^{1-\sigma }\Vert
(1-\Delta )^{\frac{k}{2}+(\frac{11}{10}+2\delta )\frac{1}{\sigma }}\alpha
\Vert _{L^{2}}^{\sigma } \\
& \lesssim \Vert \alpha \Vert _{W^{k,10}}^{1-\sigma }\Vert \alpha \Vert
_{X}^{\sigma }
\end{split}%
\end{equation*}
provided that $k>11/(5\sigma )=\frac{22}{3}+\varepsilon .$

In the case $|\eta |>|\xi -\eta |$ we proceed similarly with $f=\frac{\chi }{%
\langle \xi -\eta \rangle ^{1/2}}$. We therefore conclude the Proposition
for $\Phi _{1}.$
\end{proof}

\medskip

We now have completed the proof for $j=1$ by Lemma and Theorem. To establish %
\eqref{EstimL10NormForAlpha} for $j\neq 1,$ we note that the proposition is
clearly valid for $\Phi _{2}$ because the proof in Case 1 shows that
Proposition \ref{EstimGenPhase} is also valid in this easier case (indeed, $%
\vert\Phi_2\vert\gtrsim \max(\vert\xi\vert,\vert\xi-\eta\vert,\vert\eta\vert)
$). For $\Phi _{4},$ we note $\Phi _{4}(\xi ,\eta )=-\Phi _{1}(\eta ,\xi ),$
and repeat the same proof in light of Proposition \ref{EstimGenPhase}.
Finally, for $\Phi _{3}(\xi ,\eta )=-\Phi _{1}(\xi -\eta ,\xi ),$ we make a
change of integration variable $\eta \rightarrow \xi -\eta $ in the
integrations in both the cubic terms and $\mathfrak{B}$ and get back to the
previous case. We thus conclude the proof.

\subsection{End of the proof}

Now, we are ready to finish the proof of Theorem \ref{MainThm}.

\begin{proof}[Proof of Theorem \protect\ref{MainThm}]
The existence of a local regular solution $\beta\in C(0,T^\ast),X)$ follows
from the standard method of Kato \cite{Kat}. Combining Proposition \ref%
{ControlL2NormProp} and Proposition \ref{EstimL10NormForAlphaProp}, we
obtain that 
\begin{equation*}
\Vert \beta\Vert_X\lesssim \Vert \alpha(0)\Vert_{Y}+\Vert \beta\Vert_X^2
\end{equation*}
so that if $\Vert\alpha(0)\Vert_Y$ is sufficiently small, we get a global
bound on the $X$-norm of the solution, which implies that $T^\ast=\infty$
and gives a global bound on the $X$-norm of $\rho$ and $v$. This ends the
proof.
\end{proof}

%%%%%%%%%%%%%%%%%%%%%%%%%%%%%%%%%%%%%%%%%%%%%%%%%%%%%%%%%%%%%%%%%%%%%%%%%%%%%%%%%%%%%%%%%%%%%%%%%%%%%%%%%%%%%%%%%%%%%%%%%%%%%%%%%%%%%%%%%%%%%%%%%%%%%%%%%%%%%%%%%%%%%%%%%%%%%%%%%%%%%%%%%%%%%%%%%%%%%%%%%%%%%%%%%%%%%%%%%%%%%%%%%%%%%%%%%%%%%%%%%%%%%%%%%%%%%%%%%%%%%%%%%%%%%%%%%%%%%%%%%%%%%%%%%%%%%%%%%%%%%%%%%%%%%%%%%%

\end{document}